%% file: residual_ratio_thresholding.tex
\documentclass[english,10pt,journal]{IEEEtran}
\usepackage{color,graphicx,amsmath,amssymb,amsthm,epsfig,mathrsfs,cite,bm,enumerate}
\usepackage{array}
\usepackage{verbatim}
\usepackage{mathtools}

\usepackage{color,soul}
\usepackage{multirow}
\makeatletter
\newcommand*{\rom}[1]{\expandafter\@slowromancap\romannumeral #1@}
\makeatother
\newcommand{\squeezeup}{\vspace{-2.5mm}}
\include{notation}
\theoremstyle{remark} \newtheorem{remark}{Remark}
\usepackage{amssymb}
\usepackage{float}
\usepackage{tikz}
\usetikzlibrary{trees}
\usepackage[lofdepth,lotdepth]{subfig}
\usepackage{enumerate}
\newtheorem{theorem}{Theorem}
\newtheorem{corollary}[theorem]{Corollary}
\usepackage{multicol}
\usepackage{stfloats}
\title{Residual Ratio Thresholding for Model Order Selection } 

\author{Sreejith Kallummil,  \hspace{0cm} Sheetal Kalyani  \\
 Department of Electrical Engineering \\  Indian Institute of Technology Madras\\
  Chennai, India 600036 \\
  \{ee12d032,skalyani\}@ee.iitm.ac.in
  }

\begin{document}
 \maketitle
\begin{abstract}
Model order selection (MOS) in linear regression models is a widely studied problem in signal processing. Techniques based on information theoretic criteria (ITC) are algorithms of choice in MOS problems.  This article proposes a novel technique called residual ratio thresholding for MOS in linear regression models which is fundamentally different from the  ITC  based MOS criteria widely discussed in literature. This article also provides a rigorous mathematical analysis of the high signal to noise ratio (SNR) and large sample size behaviour of RRT. RRT is numerically shown to deliver a highly competitive performance when compared to popular model order selection criteria like Akaike information criterion (AIC), Bayesian information criterion (BIC), penalised adaptive likelihood (PAL) etc. especially when the sample size is small.
\end{abstract}

\section{Introduction}
Consider a linear regression model ${\bf y}={\bf X}\boldsymbol{\beta}+{\bf w}$, where ${\bf X}=[{\bf x}_1,\dotsc,{\bf x}_p] \in \mathbb{R}^{n \times p}$ is a known design matrix with columns $\{{\bf x}_k\}_{k=1}^p$, $\boldsymbol{\beta}\in \mathbb{R}^p$ is an unknown regression vector and ${\bf w}$ is a Gaussian distributed noise vector with  mean ${\bf 0}_n$ and covariance matrix $\sigma^2{\bf I}_n$. Here ${\bf 0}_n$ is the $n \times 1$ zero vector and ${\bf I}_n$ is the $n\times n$ identity matrix. We assume that the design matrix ${\bf X}$ has full column rank, i.e., $rank({\bf X})=p$ which is possible only if $n \geq p$. The noise variance $\sigma^2$ is assumed to be unknown.   Model order of the regression vector $k_0$ is the last index $k$ such that $\boldsymbol{\beta}_j=0$ for all $j>k$. Mathematically,  $k_0=\max\{k:\boldsymbol{\beta}_k\neq 0\}$ or equivalently   $k_0=\min\{k:\boldsymbol{\beta}_j= 0,\forall j> k\}$. In many cases of practical interest, $\boldsymbol{\beta}_k\neq 0$ for all $k\leq k_0$. In those situations, model order also corresponds to the number of non-zero entries in the regression vector $\boldsymbol{\beta}$. We also assume that the regression vector $\boldsymbol{\beta}\neq {\bf 0}_p$ which ensures that $k_0\geq 1$.  Model order selection (MOS)\cite{stoica2004model}, i.e., identification or detection of model order $k_0$ using ${\bf y}$ and ${\bf X}$ has got many applications including channel estimation in wireless communications\cite{raghavendra2005improving,tomasoni2013efficient}, fixing filter lengths in digital signal processing\cite{filter_design}, fixing the order in auto regressive (AR) time series models\cite{schmidt2011estimating} etc. This article deals with the development of novel techniques for MOS. After presenting the notations used in this article, we discuss the prior art on MOS and the  novel contributions in this article. 
\subsection{Notations used}
  Bold upper case letters represent matrices and bold lower case letters represent vectors.  $span({\bf X})$ is the column space of ${\bf X}$. ${\bf X}^T$ is the transpose and ${\bf X}^{\dagger}=({\bf X}^T{\bf X})^{-1}{\bf X}^T$ is the   pseudo inverse of ${\bf X}$. $[k]$ denotes the set $\{1,2,\dotsc,k\}$. ${\bf X}_{\mathcal{J}}$ denotes the sub-matrix of ${\bf X}$ formed using  the columns indexed by $\mathcal{J}$. In particular ${\bf X}_{[k]}=[{\bf x}_1,\dotsc,{\bf x}_k]$. ${\bf P}_{k}={\bf X}_{[k]}{\bf X}^{\dagger}_{[k]}$ is the projection matrix onto  $span({\bf X}_{[k]})$.     ${\bf a}_{\mathcal{J}}$ and ${\bf a}({\mathcal{J}})$ both denote the  entries of vector ${\bf a}$ indexed by $\mathcal{J}$.   $\|{\bf a}\|_q=(\sum\limits_{j=1}^m|{\bf a}_j|^q)^{1/q} $ is the $l_q$ norm of ${\bf a}\in \mathbb{R}^m$.  $\phi$ represents the null set. ${\bf O}_n$ is the $n\times n$ zero matrix. For any two index sets $\mathcal{J}_1$ and $\mathcal{J}_2$, the set difference  $\mathcal{J}_1/\mathcal{J}_2=\{j:j \in \mathcal{J}_1\& j\notin  \mathcal{J}_2\}$.  $f(m)=O(g(m))$ iff $\underset{m \rightarrow \infty}{\lim}\frac{f(m)}{g(m)}<\infty$. $\mathbb{E}(Z)$ represents the expectation of random variable/vector (R.V) $Z$ and $\mathbb{P}(\mathcal{A})$ represents the probability of event $\mathcal{A}$. ${\bf a}\sim \mathcal{N}({\bf u},{\bf C})$ implies that ${\bf a}$ is a Gaussian  R.V with mean ${\bf u}$ and covariance matrix ${\bf C}$.  $\mathbb{B}(a,b)$ denotes a Beta R.V with parameters $a$ and $b$. $B(a,b)=\int_{t=0}^1t^{a-1}(1-t)^{b-1}dt$ is the beta function with parameters $a$ and $b$. $\chi^2_k$ is a central chi square R.V with $k$ degrees of freedom (d.o.f), whereas, $\chi^2_k(\lambda)$ is  a non-central chi square R.V with $k$ d.o.f and non-centrality parameter $\lambda$. $G(x)=\int_{t=0}^{\infty}e^{-t}t^{x-1}dt$ is the Gamma function. A R.V $Z$ converges in probability to a constant $c$ (i.e., $Z\overset{P}{\rightarrow } c$) as $n \rightarrow \infty$ (or as $\sigma^2 \rightarrow 0$) if $\underset{n \rightarrow \infty}{\lim}\mathbb{P}(|Z-c|> \epsilon)=0$ (or $\underset{\sigma^2 \rightarrow 0}{\lim}\mathbb{P}(|Z-c|> \epsilon)=0$) for every fixed $\epsilon>0$.
\subsection{Prior art on MOS}
MOS is one of the most widely studied topics in signal processing.  Among the plethora of MOS techniques discussed in literature, the ones based on  information theoretic criteria (ITC) are the most popular.  The operational form of  ITC based MOS techniques  is given by
\begin{equation}\label{itc}
\hat{k_0}=\underset{k=1,\dotsc,p}{\min}n\log(\sigma^2_k)+h(k,\sigma^2_k),
\end{equation} 
where $\sigma^2_k=\|({\bf I}_n-{\bf P}_k){\bf y}\|_2^2/n$ is the maximum likelihood estimate of $\sigma^2$ assuming that the model ${\bf y}={\bf X}_{[k]}\boldsymbol{\beta}_{[k]}+{\bf w}$ is true.  The term $n\log(\sigma^2_k)$ measures how well the observation ${\bf y}$ is approximated using the columns in ${\bf X}_{[k]}$. Since one  can approximate ${\bf y}$ better using more number of columns, the term $n\log(\sigma^2_k)$ is a decreasing function of $k$. Hence,  $\underset{k=1,\dotsc,p}{\arg\min}\ n\log(\sigma^2_k)$ always equals the maximum possible model order, i.e., $p$.  The second term  $h(k,\sigma^2_k)$ in (\ref{itc}) popularly called penalty function is   typically an increasing   function of $k$. Consequently, the ITC in (\ref{itc}) select as model order estimate that $k$ which provides a good trade-off between data fit represented by $n\log(\sigma^2_k)$ and  complexity represented by $h(k,\sigma^2_k)$. Please note that the term $n\log(\sigma^2_k)$ is two times the negative of log likelihood of data ${\bf y}$ maximized w.r.t the unknown parameters $\{\boldsymbol{\beta}_{[k]},\sigma^2\}$. Consequently, (\ref{itc}) has an interpretation of minimizing penalised log likelihood, a widely popular statistical concept.  Another interesting interpretation of (\ref{itc}) from the perspective of sequential hypothesis testing is  derived in \cite{stoica2004information}.

The properties of ITC is completely determined by the penalty function $h(k,\sigma^2_k)$ and different penalty functions give different ITC.   The penalty function $h(k,\sigma^2_k)$ in (\ref{itc}) can either be a deterministic function of  $k$ as in Akaike information criteria \textit{aka} AIC ($h(k,\sigma^2_k)=2k$), large sample version of Bayesian information criteria \textit{aka} BIC  ($h(k,\sigma^2_k)=k\log(n)$)\cite{stoica2004model} etc. or a stochastic function as  in penalized adaptive likelihood PAL\cite{stoica2013model}, normalised minimum description length NMDL\cite{rissanen2000mdl}, finite sample  forms of BIC\cite{stoica2012proper}, empirical BIC\cite{nielsen2013bayesian}, exponentially embedded families (EEF)\cite{EEFPDF,ding2011inconsistency} etc. Penalty functions in popular ITC like AIC, BIC, MDL, NMDL, EEF etc.  are derived using statistical concepts like Kullbeck Leibler divergence, Laplace approximation for integrals, information theoretic complexity, exponential family of distributions etc.   Please see \cite{tsp} for a list of popular penalty functions.   Most of the analytical results on ITC  are based on either the large sample  asymptotics, i.e., $n \rightarrow \infty, p/n \rightarrow 0$ \cite{nishii1988maximum,asymptotic_map,rao1989strongly,shao1997asymptotic,zheng1995consistent,minimal} or the high signal to noise ratio (SNR) asymptotics, i.e.,  as $\sigma^2\rightarrow 0$\cite{ding2011inconsistency,tsp,schmidt2012consistency,stoica2012proper}. These asymptotic results are summarized in the following lemma. 
 \begin{lemma} ITC based MOS estimate in  (\ref{itc}) satisfies the following consistency results.\cite{nishii1988maximum,tsp}\\
 a). Suppose that $h(k,\sigma^2_k)=vk$ and maximum model order $p$ is fixed. Then $v/\log(\log(n))\rightarrow \infty$ and $v/n\rightarrow 0$  is sufficient for the large sample consistency, i.e., the probability of correct selection $PCS=\mathbb{P}(\hat{k}_0=k_0)\rightarrow 1$ as $n \rightarrow \infty$. \\
b). ITC with  $h(k,\sigma^2_k)=(ak+b)\log(1/\sigma^2_k)$ in  (\ref{itc}) is high SNR consistent (i.e., $PCS\rightarrow 1$ as $\sigma^2\rightarrow 0$) if $ak_0+b<n$.
 \end{lemma} 
Using these consistency results, it is easy to show that ITC like NMDL, proper forms of BIC, EEF etc. are high SNR consistent. Similarly, one can show that BIC, NMDL etc. are also large sample consistent.  Techniques to create novel penalty functions based on the high SNR behaviour of ITC were proposed in \cite{designITC, tsp, stoica2013model}. Please note that ITC based MOS rules are also developed for non-linear order selection problems like source number enumeration in \cite{lu2013generalized,haddadi2010statistical},  sinusoidal enumeration\cite{nielsen2014model,asymptotic_map} etc.  

\subsection{Contribution of this article}
This article proposes a novel technique for MOS called residual ratio thresholding (RRT). RRT  is based on the behaviour of adjacent residual norm ratio $RR(k)=\dfrac{\sigma^2_k}{\sigma^2_{k-1}}$ and is structurally different from the ITC based MOS criteria in (\ref{itc}). { Unlike popular algorithms like BIC, AIC etc. which are motivated by large sample asymptotics,  RRT is motivated by a finite sample distributional result (see Theorem \ref{thm:beta_mos}). This finite sample nature of RRT is reflected in its superior empirical performance compared to AIC, BIC etc. when the sample size is small. }  RRT involves a tuning parameter $\alpha$ for which we give a proper semantic interpretation using high SNR and large sample analysis. In particular, one can set the hyper parameters in RRT so as  to achieve a predetermined high SNR and large sample lower bound on PCS. In this sense, RRT is similar to the ITC design technique in \cite{designITC}. However, for the same high SNR error bound, RRT is numerically shown to deliver better PCS in the low to medium SNR regime than  \cite{designITC}.   Further, the conditions required for the large sample consistency of RRT is also derived.  { Numerical simulations indicate that RRT performs better than  existing ITC based MOS techniques in many situations including but not limited to the case of small $n$ and $p$.  In  situations where RRT is  outperformed by other MOS criteria,  RRT performed close to the best performing MOS criterion.} Based on the derived analytical results and observed numerical results, we believe that RRT deserves a place in the algorithmic toolkit for MOS problems. 

This article is organized as follows. Section \rom{2} analyses the behaviour of residual ratio $RR(k)$. Section \rom{3} presents and analyses the RRT based MOS technique. Section \rom{4} presents numerical simulations.
\section{Behaviour of residual ratios}
Define ${\bf r}^k=({\bf I}_n-{\bf P}_k){\bf y}$, the residual after projecting onto the column space of ${\bf X}_{[k]}$. In terms of $\sigma^2_k$, $\|{\bf r}^k\|_2^2=n\sigma^2_k$.  In this section, we rigorously analyse the behaviour of residual ratios $RR(k)=\dfrac{\|{\bf r}^k\|_2^2}{\|{\bf r}^{k-1}\|_2^2}=\dfrac{\sigma^2_k}{\sigma^2_{k-1}}=\dfrac{\|({\bf I}_n-{\bf P}_k){\bf y}\|_2^2}{\|({\bf I}_n-{\bf P}_{k-1}){\bf y}\|_2^2}$ for $k\geq k_0$. The proposed RRT technique for MOS is based on this analysis. The basic distributional results are listed in the following lemma.
\begin{lemma}\label{lemma:basic_distributions} $RR(k)$ satisfies the following  for all $\sigma^2>0$.\\
a). $RR(k)$ for $k>k_0$ satisfies $RR(k)\sim \mathbb{B}(\dfrac{n-k}{2},\dfrac{1}{2})$.\\
b). $RR(k_0)=\dfrac{Z_1}{Z_1+Z_2}$, where $Z_1=\|({\bf I}_n-{\bf P}_{k_0}){\bf w}\|_2^2\sim \sigma^2\chi^2_{n-k_0}$ and $Z_2=\|({\bf P}_{k_0}-{\bf P}_{k_0-1}){\bf y}\|_2^2\sim \sigma^2\chi^2_1\left(\dfrac{\|({\bf I}_{n}-{\bf P}_{k_0-1}){\bf x}_{k_0}\|_2^2\boldsymbol{\beta}_{k_0}^2}{\sigma^2}\right)$.
\end{lemma} 
\begin{proof}
Please see Appendix A. 
\end{proof}

We now give a bound in Theorem \ref{thm:beta_mos} which is a direct consequence of the distributional result a) in  Lemma \ref{lemma:basic_distributions}.
\begin{thm}\label{thm:beta_mos} 
Define $\Gamma_{RRT}^{\alpha}(k)=F_{\frac{n-k}{2},\frac{1}{2}}^{-1}\left(\frac{\alpha}{p}\right)$ for $1\leq k\leq p<n$, where $F_{\frac{n-k}{2},\frac{1}{2}}^{-1}()$  is the inverse cumulative distribution function (CDF) of a $\mathbb{B}(\frac{n-k}{2},\frac{1}{2})$ R.V. Then $\mathbb{P}(RR(k)>\Gamma_{RRT}^{\alpha}(k),\forall k> k_0)\geq 1-\alpha$ for each $0<\alpha<1$ and $\sigma^2>0$.
\end{thm}
\begin{proof} The proof follows directly from the union bound\footnote{For any $n$ events $\{\mathcal{A}_i\}_{i=1}^n$, union bound is $\mathbb{P}(\cup_{i=1}^nA_i)\leq \sum \limits_{i=1}^n\mathbb{P}(A_i)$}, the definition of $\Gamma_{RRT}^{\alpha}(k)$ and the result $RR(k)=\dfrac{\|{\bf r}^k\|_2^2/\sigma^2}{\|{\bf r}^{k-1}\|_2^2/\sigma^2}\sim  \mathbb{B}(\dfrac{n-k}{2},\dfrac{1}{2})$ for $k>k_0$. 
\begin{equation}
\begin{array}{ll}
\mathbb{P}(RR(k)>\Gamma_{RRT}^{\alpha}(k),\forall k> k_0)\\
\ \ \ \ \ \ =1-\mathbb{P}(\exists k>k_0: RR(k)<\Gamma_{RRT}^{\alpha}(k))\\
\ \ \ \ \ \ \geq 1-\sum\limits_{k>k_0}^p\mathbb{P}(RR(k)<\Gamma_{RRT}^{\alpha}(k))\\
\ \ \ \ \ \ =1-\sum\limits_{k>k_0}^pF_{\frac{n-k}{2},\frac{1}{2}}\left(F_{\frac{n-k}{2},\frac{1}{2}}^{-1}\left(\frac{\alpha}{p}\right)\right)\\
\ \ \ \ \ \ = 1-\frac{(p-k_0)}{p}\alpha\geq 1-\alpha. 
\end{array}
\end{equation}
\end{proof}
Theorem \ref{thm:beta_mos} implies that $RR(k)$ for $k>k_0$ is lower bounded by $\Gamma_{RRT}^{\alpha}(k)$ with a very high probability (for small values of $\alpha$).   Please note that the bound in Theorem  \ref{thm:beta_mos} hold true irrespective of the value of $k_0$. Also note that the  lower bound $\Gamma_{RRT}^{\alpha}(k)$ on $RR(k)$ for $k>k_0$ is independent of $\sigma^2$. Certain other interesting properties of $\Gamma_{RRT}^{\alpha}(k)$ is listed below.
\begin{lemma} $\Gamma_{RRT}^{\alpha}(k)$ satisfies the following properties. \\
a). For fixed $n$ and $p$, $\Gamma_{RRT}^{\alpha}(k)$ decreases with decreasing $\alpha$. In particular $\Gamma_{RRT}^0(k)=0$ and $\Gamma_{RRT}^p(k)=1$.\\
b). For fixed $n$ and $\alpha$, $\Gamma_{RRT}^{\alpha}(k)$ decreases with increasing $p$.\\
c). For fixed $n$, $p$ and $\alpha$, $\Gamma_{RRT}^{\alpha}(k)$ decreases with increasing $k$.
\end{lemma}
\begin{proof}
a) and b) follow from the monotonicity of CDF and the fact that a Beta distribution has support only in [0,1].  c) is true since the Beta CDF $F_{a,b}(x)$ is a decreasing function of $a$ for fixed values of $b$ and inverse CDF $F^{-1}_{a,b}(x)$ is an increasing function of $a$ for fixed values of $b$.
\end{proof}

We next consider the behaviour of $RR(k_0)$ as $\sigma^2 \rightarrow 0$. The main result is stated in the following Theorem. 
\begin{thm}\label{thm:rrknot} $RR(k_0)\overset{P}{\rightarrow}0$ as $\sigma^2 \rightarrow 0$.
\end{thm}
\begin{proof}Please see Appendix B.
\end{proof}
Theorem \ref{thm:rrknot} implies that $RR(k_0)$ takes smaller and smaller values with increasing SNR. This is in contrast with $RR(k)$ for $k>k_0$ which is lower bounded  by a constant independent of the operating SNR. { The analysis of $RR(k)$ for $k<k_0$ is not relevant to the RRT algorithm discussed in this article and is thereby omitted. However, following the proof of Theorem \ref{thm:rrknot}, one can easily show that $RR(k)$ for $k<k_0$ converges in probability to a constant $c_k=\dfrac{\|({\bf I}_n-{\bf P}_k){\bf X}\boldsymbol{\beta}\|_2^2}{\|({\bf I}_n-{\bf P}_k){\bf X}\boldsymbol{\beta}\|_2^2+\|({\bf P}_{k}-{\bf P}_{k-1}){\bf X}\boldsymbol{\beta}\|_2^2}$ which is strictly bounded away from zero and one. }

\subsection{Numerical Validation }
\begin{figure*}[htb]
\begin{multicols}{2}

    \includegraphics[width=1\linewidth]{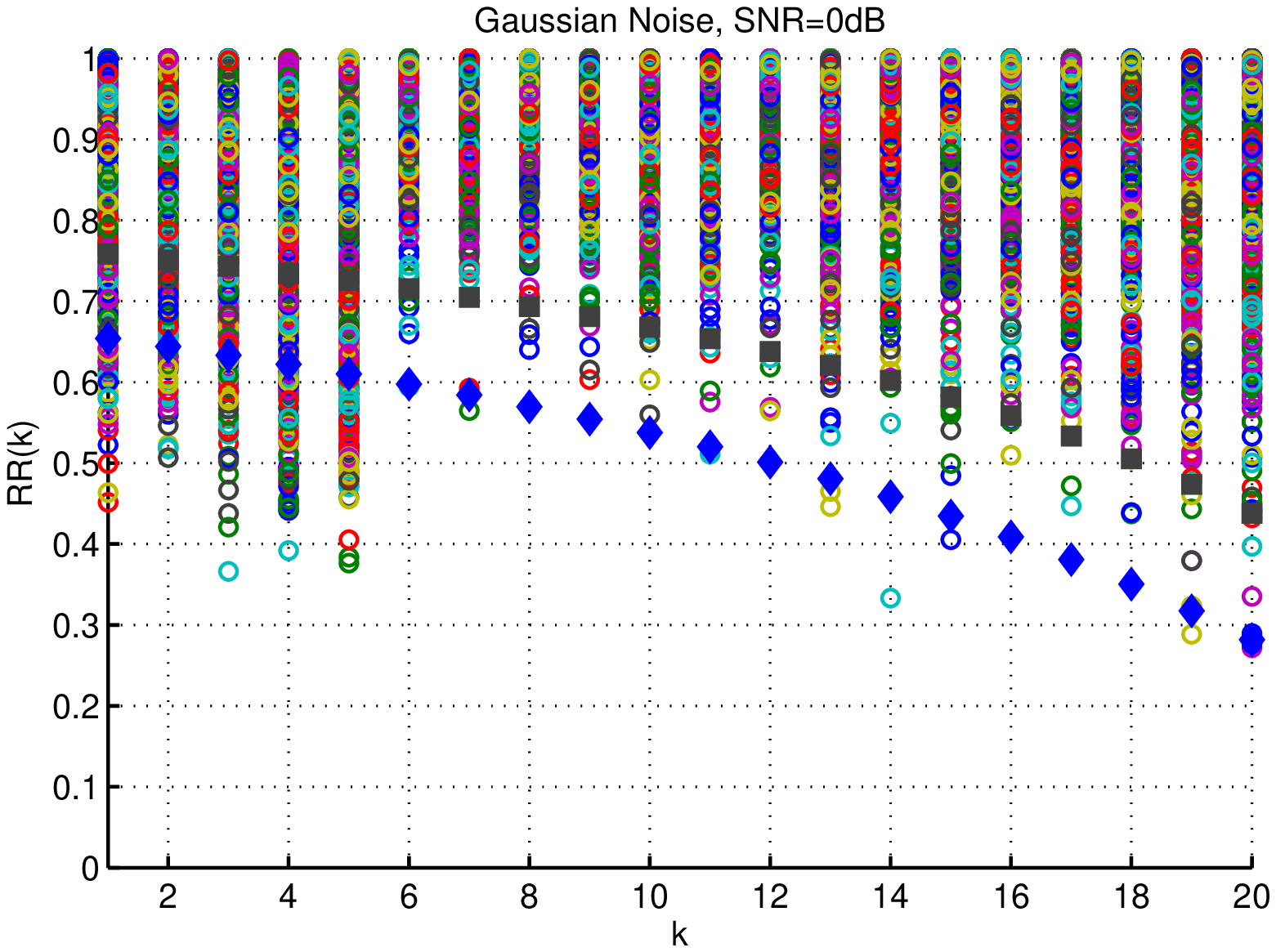} 
    \caption*{a). SNR=0dB. $\{RR(k)<\Gamma_{RRT}^{\alpha}(k),\forall\ k>k_{0}\}$ $93.5\%$ for ($\alpha=0.1$), $99\%$ for ($\alpha=0.01$) }
    
    \includegraphics[width=1\linewidth]{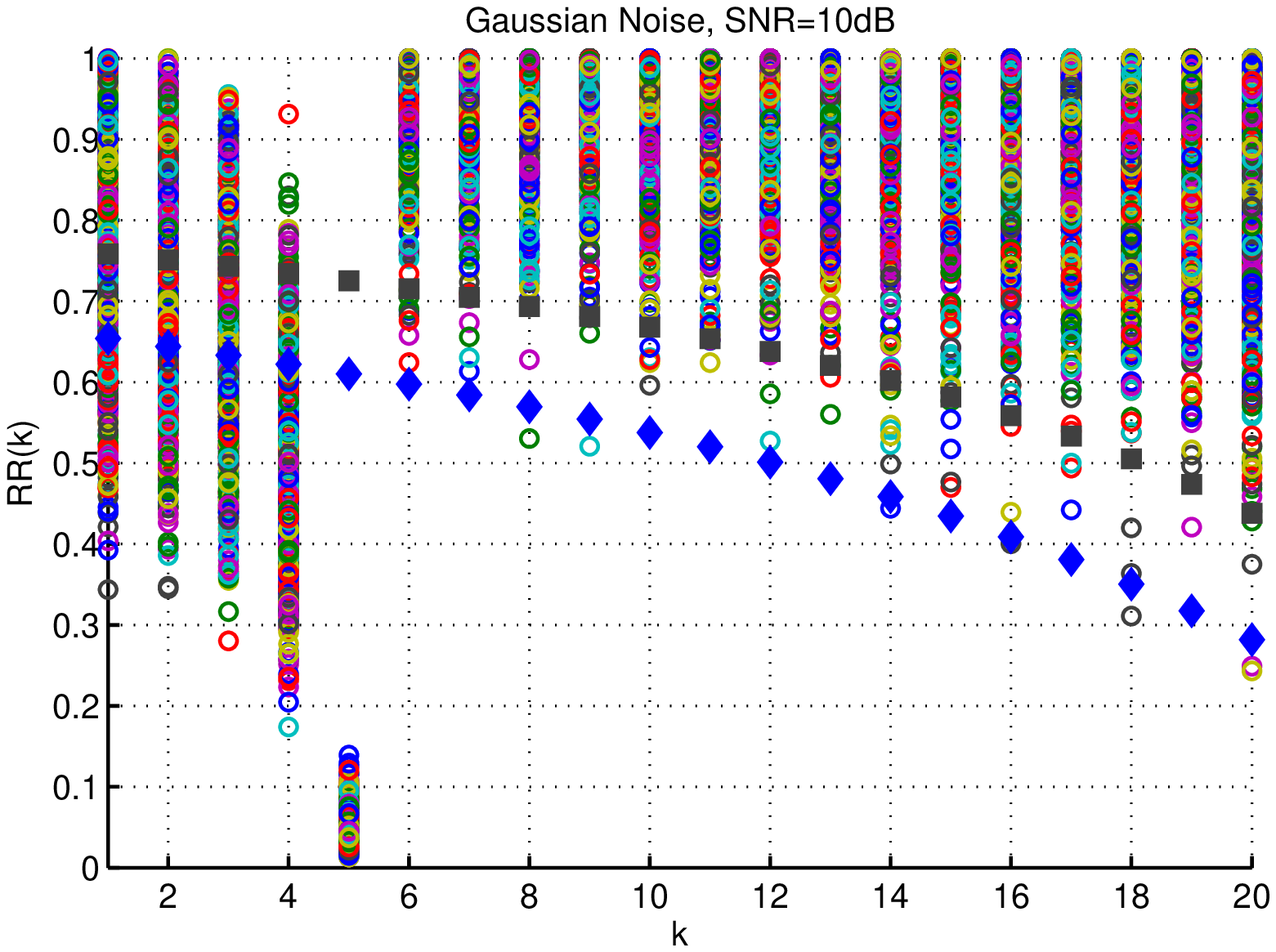}
    \caption*{b). SNR=20dB.  $\{RR(k)>\Gamma_{RRT}^{\alpha}(k),\forall\ k>k_{0}\}$ $94.1\%$ for ($\alpha=0.1$), $99.3\%$ for ($\alpha=0.01$) }
   \end{multicols} 
   \squeezeup
   \caption{Behaviour of $RR(k)$. $k_0=5$.  SNR=0dB (left) and SNR=20dB (right). Circles in Fig.1 represent the values of $RR(k)$, squares represent $\Gamma_{RRT}^{\alpha}(k)$ with $\alpha=0.1$ and diamonds represent $\Gamma_{RRT}^{\alpha}(k)$ with $\alpha=0.01$.  }
   \label{fig:evolution}
   \squeezeup
\end{figure*}
We next numerically validate the distributional results derived in previous subsections, \textit{viz}.  Theorem \ref{thm:beta_mos} and Theorem \ref{thm:rrknot}.  Consider a $30 \times 20$ design matrix ${\bf X}$ generated using independent $\mathcal{N}(0,1/n)$ entries.   $k_0$ is set at $k_0=5$ and $\boldsymbol{\beta}_k=\pm 1$ for all $k\leq k_0$. We plot 1000 realizations of $\{RR(k)\}_{k=1}^{p}$ at two different SNRs, \textit{viz} SNR=0dB (Fig.1.a) and SNR=20dB (Fig.1.b). From these plots and the empirically evaluated probabilities of $\{RR(k)>\Gamma_{RRT}^{\alpha}(k),\forall k>k_0\}$ reported alongside, it is clear that the $1-\alpha$ probability bound predicted by Theorem \ref{thm:beta_mos} holds true. Further, as one can see from Fig.1.a and Fig.1.b, the value of $RR(k_0)=RR(5)$ decreases with increasing SNR. This observation is in agreement  with the convergence result $RR(k_0)\overset{P}{\rightarrow} 0$ as $\sigma^2\rightarrow 0$ in Theorem \ref{thm:rrknot}.

\section{Residual ratio thresholding based MOS}
From the behaviour of $RR(k)$ discussed analytically and numerically in section \rom{2}, it is clear that $RR(k)$ for $k>k_0$ is larger than $\Gamma_{RRT}^{\alpha}(k)$ with a very high probability (for smaller values of $\alpha$), whereas, $RR(k_0)$ decreases to zero with increasing SNR or equivalently decreasing $\sigma^2$. Consequently, $RR(k_0)$ will be smaller than $\Gamma_{RRT}^{\alpha}(k_0)$ at high SNR, whereas, $RR(k)>\Gamma_{RRT}^{\alpha}(k)$ for all $k>k_0$ with a high probability. Hence, with increasing SNR, the model order estimate
\begin{equation}\label{rrt}
\hat{k}_{RRT}=\max\{k:RR(k)\leq \Gamma_{RRT}^{\alpha}(k)\}
\end{equation}
will corresponds to $k_0$ with a very high probability.   This is the RRT based MOS criterion  proposed  in this article. The efficacy of RRT is visible from Fig.1.b itself where $\hat{k}_{RRT}=k_0$ with probability $94.1\%$ for $\alpha=0.1$ and probability $99.3\%$ for $\alpha=0.01$ respectively.
\begin{remark} \label{rem:substitute}
 An important aspect regarding the RRT based MOS in (\ref{rrt}) is the choice of $\hat{k}_{RRT}$ when the set $\{k:RR(k)<\Gamma_{RRT}^{\alpha}(k)\}=\phi$. This situation happens only at very low SNR. Further, throughout this article, we assumed that $k_0\geq 1$. Hence, setting $\hat{k}_{RRT}=0$ when $\{k:RR(k)<\Gamma_{RRT}^{\alpha}(k)\}=\phi$ is not a prudent choice.  In this article, we set $\hat{k}_{RRT}=\max\{k:RR(k)\leq \Gamma_{RRT}^{\alpha_{new}}(k)\}$ where
\begin{equation}\label{alphanew}
\alpha_{new}=\underset{a>\alpha}{\min} \{a: \{k:RR(k)\leq \Gamma_{RRT}^{\alpha}(k)\}\neq\phi\}.
\end{equation}
Since $\alpha=p$ gives $\Gamma_{RRT}^{\alpha}(1)=1$ and $RR(1)\leq 1$,   a value of $\alpha_{new}\leq p$ always exist. $\alpha_{new}$ can be easily computed by first pre-computing  $\{\Gamma_{RRT}^{a}(k)\}_{k=1}^p$ for say 100 prefixed values of $a$ in the interval $(\alpha,p]$.   Note that the value of $\alpha_{new}$ can be greater than 1 and hence $\alpha_{new}$ does not have any probabilistic interpretation. Please note that $\Gamma_{RRT}^{\alpha}(k)$ and $\{\Gamma_{RRT}^{a}(k)\}_{k=1}^p$ in $(\alpha,p]$ can all be precomputed. Hence, the online computational complexity of RRT is same as that of ITC in (\ref{itc}).
\end{remark}
\begin{remark} RRT is directly based  on the evolution of residual norms and residual ratios with increasing SNR. This is in contrast with AIC, BIC etc. whose penalty terms  are based on information theoretic arguments and their asymptotic approximations. This is a fundamental philosophical difference between AIC, BIC etc. and RRT. In this sense, RRT is philosophically closer to PAL\cite{stoica2013model}, whose penalty term is also derived directly from the behaviour of residual norms. { The fact that RRT is based on finite sample results directly related to the statistics involved in MOS explains the superior performance of RRT \textit{viz a viz} BIC, AIC etc. when the sample size is small (see section \rom{4}).} 
\end{remark}

\subsection{High SNR behaviour and interpretation of $\alpha$}
We next explain the high SNR behaviour of $\hat{k}_{RRT}$. Define the probabilities of overestimation and underestimation of $\hat{k}_{RRT}$ as $\mathbb{P}_{\mathcal{O}}=\mathbb{P}(\{\hat{k}_{RRT}>k_0\})$ for $k_0<p$  and  $\mathbb{P}_{\mathcal{U}}=\mathbb{P}(\{\hat{k}_{RRT}<k_0\})$ for $k_0>1$ respectively. 
\begin{thm} 
\label{thm:highSNR}Overestimation and underestimation probabilities of RRT satisfy   $\underset{\sigma^2 \rightarrow 0}{\lim }\mathbb{P}_{\mathcal{O}}\leq \alpha$ and  $\underset{\sigma^2 \rightarrow 0}{\lim }\mathbb{P}_{\mathcal{U}}=0$ respectively. Consequently, $\underset{\sigma^2 \rightarrow 0}{\lim }PCS\geq 1-\alpha$.
\end{thm}
\begin{proof} Please see Appendix C.
\end{proof}
\begin{remark} Theorem 1 gives a straight forward operational interpretation for the tuning parameter $\alpha$ as the high SNR upper bound on the probability of overestimation $\underset{\sigma^2 \rightarrow 0}{\lim }\mathbb{P}_{\mathcal{O}}\leq \alpha$ and probability of error $\underset{\sigma^2 \rightarrow 0}{\lim }1-PCS\leq \alpha$. Such a straight forward semantic interpretation is not available for the tuning parameters in AIC, BIC etc.  At all SNR where the set $\{k:RR(k)<\Gamma_{RRT}^{\alpha}(k)\}\neq \phi$, overestimation probability is given by $\mathbb{P}(\exists k>k_0: RR(k)<\Gamma_{RRT}^{\alpha}(k))$ which by Theorem \ref{thm:beta_mos} is less than $\alpha$. Further, the  probability that the set $\{k:RR(k)<\Gamma_{RRT}^{\alpha}(k)\}= \phi$ is very low at all practical SNR regimes. Hence,   the  bound $\mathbb{P}_{\mathcal{O}}\leq \alpha$ hold true even when the SNR is very low. However, the bound $1-PCS\leq \alpha$ hold true only when the SNR is very high. These observations will be numerically validated in section \rom{4}.
\end{remark}
\begin{remark} While designing the penalty function $h(k,\sigma^2_k)$ in ITC (\ref{itc}) or the parameter $\alpha$ in RRT, the user has control only over the high SNR behaviour of $\mathbb{P}_{\mathcal{O}}$. When $h(k,\sigma^2_k)$ is of the form $vk$ for some fixed parameter $v>0$ (like AIC, BIC etc.) and the user requires the high SNR $\mathbb{P}_{\mathcal{O}}$ to be lower than a predefined value $\mathbb{P}^{des}_{\mathcal{O}}$, \cite{designITC} proposed  to set $v=v^{des}$, where $v^{des}$  is  the minimum value of $v$  that delivers $\underset{\sigma^2 \rightarrow 0}{\lim}\mathbb{P}_{\mathcal{O}}\leq \mathbb{P}^{des}_{\mathcal{O}}$ assuming that $k_0=0$. The case of $k_0=0$ is worst case scenario in terms of overestimation.  To operate RRT satisfying  $\underset{\sigma^2 \rightarrow 0}{\lim}\mathbb{P}_{\mathcal{O}}\leq \mathbb{P}^{des}_{\mathcal{O}}$, one can set $\alpha=\mathbb{P}^{des}_{\mathcal{O}}$. Numerical simulations indicate that for the same value of $\mathbb{P}^{des}_{\mathcal{O}}$, RRT  very often  deliver PCS higher than that of the design criteria in \cite{designITC} in the low to moderate SNR regime.
\end{remark}
\subsection{High SNR inconsistency of RRT} 
From the $RR(k)\sim \mathbb{B}\left(\dfrac{n-k}{2},\dfrac{1}{2}\right)$ distribution in (\ref{beta_prelim}) and $\Gamma_{RRT}^{\alpha}(k)=F_{\frac{n-k}{2},\frac{1}{2}}^{-1}(\frac{\alpha}{p})$, it is true that  $\mathbb{P}(RR(k)<\Gamma_{RRT}^{\alpha}(k))=F_{\frac{n-k}{2},\frac{1}{2}}(F_{\frac{n-k}{2},\frac{1}{2}}^{-1}(\frac{\alpha}{p}))=\alpha/p$ for $k>k_0$. This implies that 
\begin{equation}
\mathbb{P}_{\mathcal{O}}\geq \mathbb{P}(\exists k>k_0: RR(k)<\Gamma_{RRT}^{\alpha}(k))\geq \alpha/p>0,\forall \sigma^2>0.
\end{equation}
Consequently, RRT with $\sigma^2$ independent values of $\alpha$ is not high SNR consistent. 
However, even in  a small scale problem with $p=10$, the lower bound on $\mathbb{P}_{\mathcal{O}}$ gives $0.01$ for $\alpha=0.1$ and $0.001$ for $\alpha=0.01$,  whereas, the upper bound gives $0.1$ and $0.01$ respectively.  Hence, for values of $\alpha$ like  $\alpha=0.1$ or $\alpha=0.01$,   the difference in PCS between a high SNR consistent MOS and RRT at high SNR would be negligible. Please note that MOS criteria like AIC, BIC, design criteria in \cite{designITC} etc. are also inconsistent at high SNR.  Further, numerical simulations indicate that RRT outperforms  high SNR consistent MOS  criteria such as \cite{tsp, stoica2013model} etc.  in the low and moderate SNR regimes. Consequently, the negligible performance loss at high SNR due to inconsistency is compensated by the good overall performance of RRT. Also please note that the high SNR performance of RRT with $\alpha=0.1$ or $\alpha=0.01$ is better than that of other high SNR inconsistent criteria like AIC, BIC etc. when the sample size is small.  
\subsection{Large sample behaviour of $\Gamma_{RRT}^{\alpha}(k_0)$}
\begin{figure*}[htb]
\begin{multicols}{3}

    \includegraphics[width=1\linewidth]{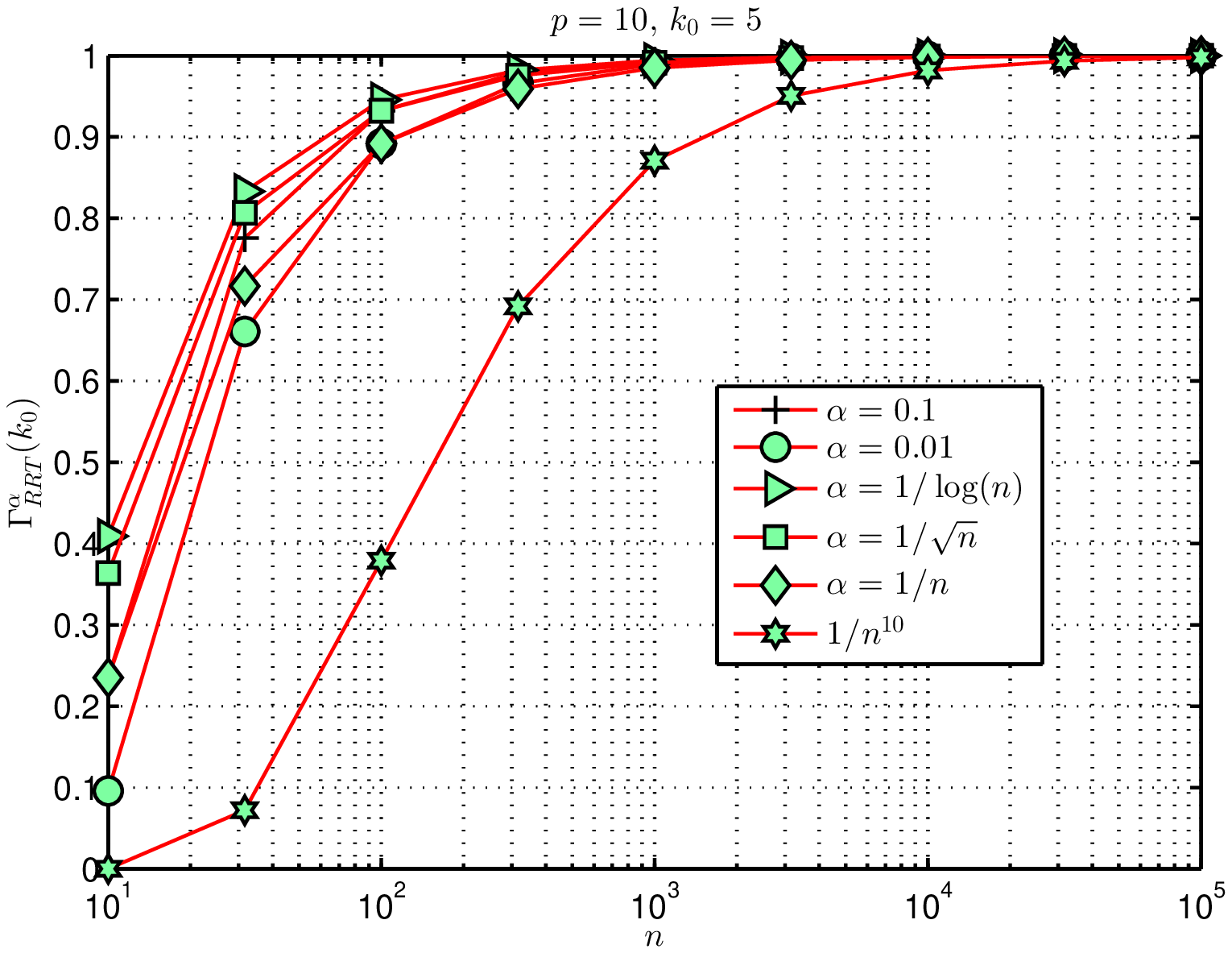} 
    \caption*{a). $p$, $k_0$ fixed. $n$ increasing.}
    
    \includegraphics[width=1\linewidth]{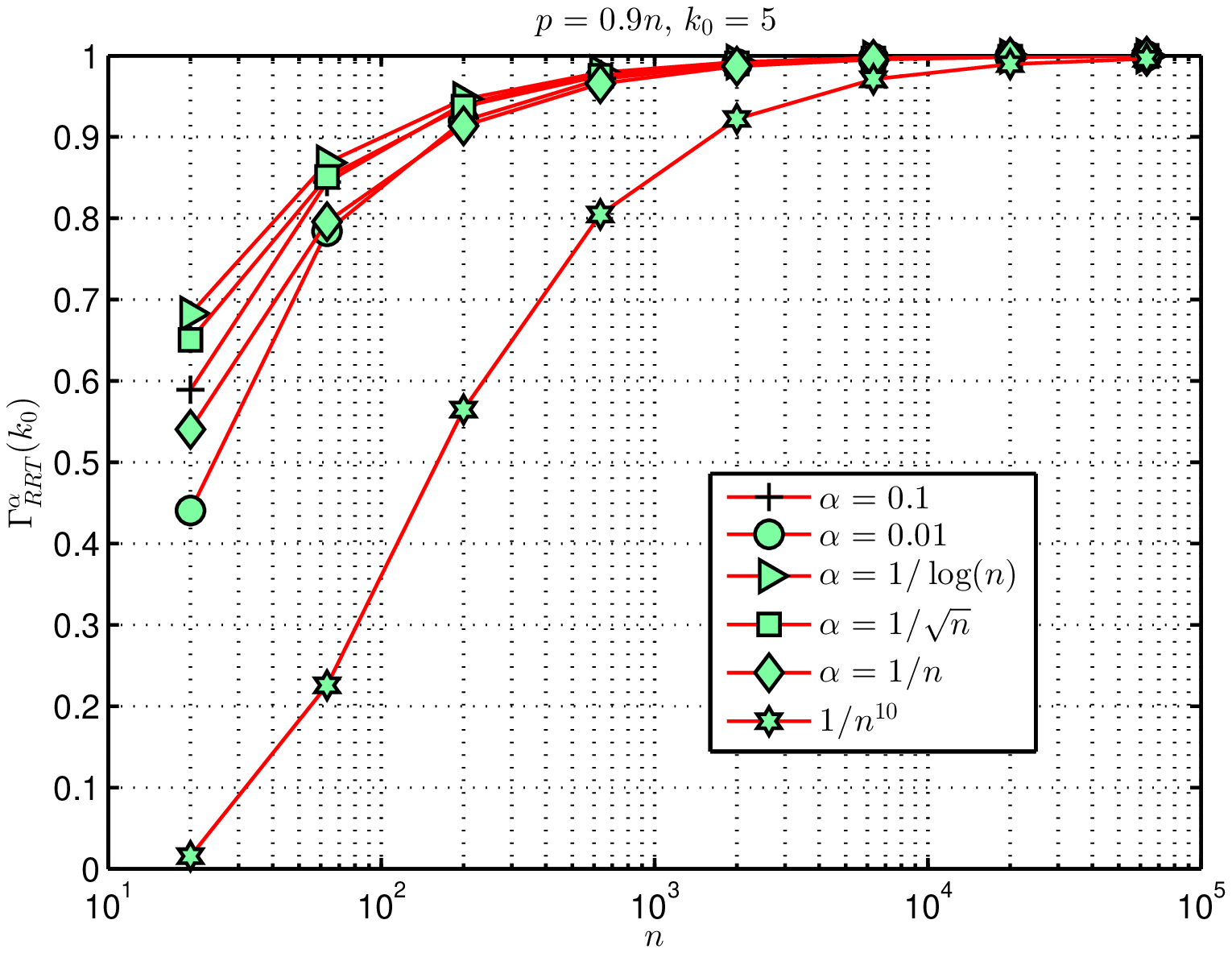} 
    \caption*{b).  $k_0$ fixed. $p=0.9n$, $n$ increasing.}
    
     \includegraphics[width=1\linewidth]{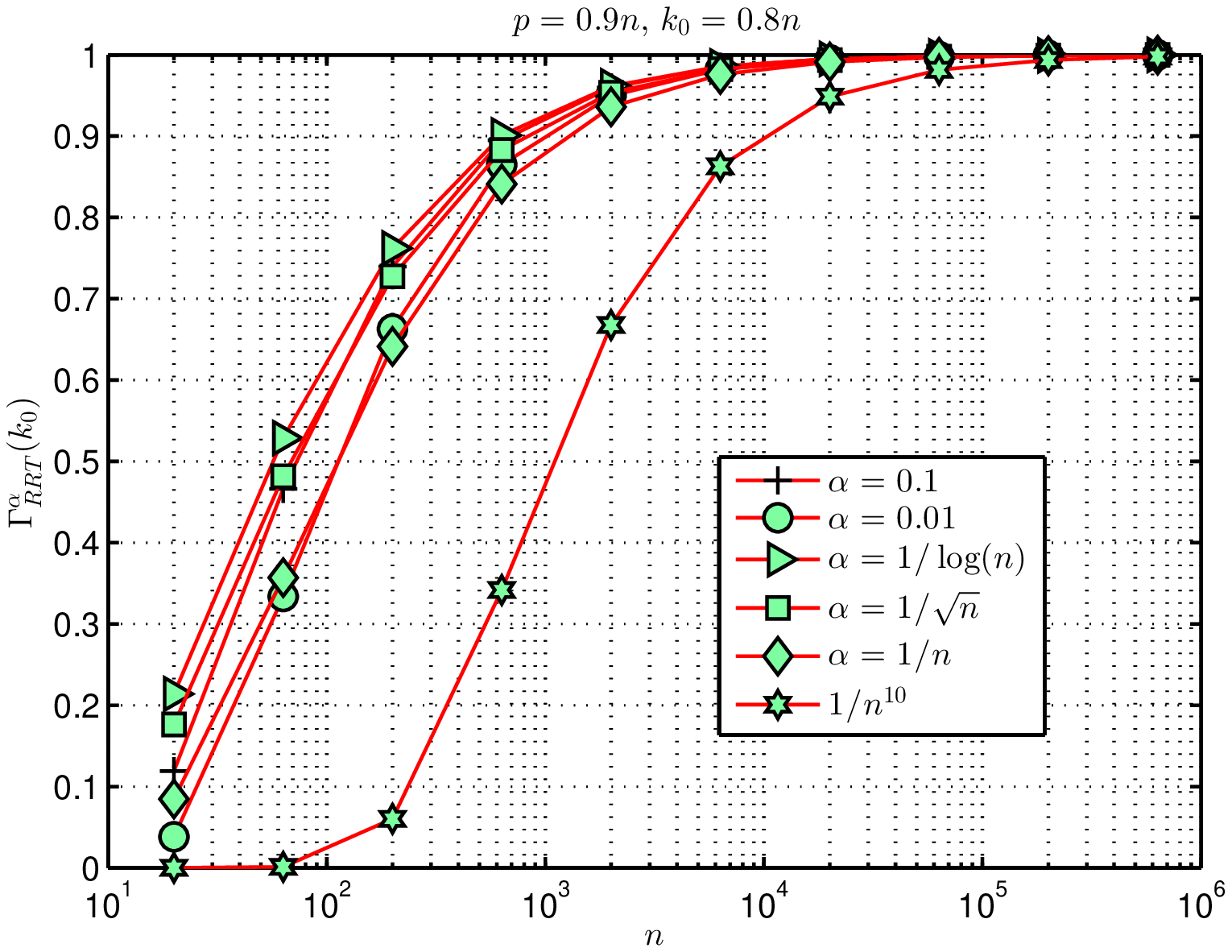} 
    \caption*{c). $k_0=0.8n$, $p=0.9n$, $n$ increasing. }
    
   \end{multicols} 
   \squeezeup
   \caption{Asymptotic behaviour of $RR(k_0)$. }
   \label{fig:asymptotic}
   \squeezeup
\end{figure*}
In the following two subsections, we evaluate the large sample behaviour of RRT. As a prelude, we first analyse the behaviour of the function $\Gamma_{RRT}^{\alpha}(k_0)$ as $n \rightarrow  \infty$. 
\begin{thm}\label{thm:asymptotic_rrt}
Let $n$ increase to $\infty$ such that $p/n\rightarrow [0,1)$ and $k_{lim}=\underset{n\rightarrow \infty}{\lim}k_0/n\in[0,1)$. Parameter $0\leq \alpha\leq 1$ is either a fixed number or a function of $n$  with limits $\underset{n \rightarrow \infty}{\lim}\alpha=0$ and  $-\infty\leq \alpha_{lim}=\underset{n\rightarrow \infty}{\lim}\log(\alpha)/n\leq 0$. Then, $\Gamma_{RRT}^{\alpha}(k_0)=F_{\frac{n-k_0}{2},\frac{1}{2}}^{-1}(\frac{\alpha}{p})$ satisfies the following asymptotic limits.\\
A1). $\underset{n \rightarrow \infty}{\lim}\Gamma_{RRT}^{\alpha}(k_0)=1$ if $\alpha_{lim}=0$.\\
A2). $0<\underset{n \rightarrow \infty}{\lim}\Gamma_{RRT}^{\alpha}(k_0)=e^{\frac{2\alpha_{lim}}{1-k_{lim}}}<1$ if $-\infty<\alpha_{lim}<0$.\\
A3). $\underset{n \rightarrow \infty}{\lim}\Gamma_{RRT}^{\alpha}(k_0)=0$ if $\alpha_{lim}=-\infty$.
\end{thm}
\begin{proof}Please see Appendix D.
\end{proof}
Theorem \ref{thm:asymptotic_rrt} implies that when $\alpha$ is reduced to zero with increasing $n$  at a rate slower than $a^{-n}$ for some $a>1$, (i.e., $\alpha_{lim}=0$), then it is possible to achieve a value of $\Gamma_{RRT}^{\alpha}(k_0)$ arbitrarily close  to one at large $n$.  Choices of $\alpha$ that satisfy $\alpha_{lim}=0$ include $\alpha=\text{constant}$, $\alpha=1/\log(n)$, $\alpha=1/n^c$ for some $c>0$ etc. However, if one decreases $\alpha$ to zero at a rate $a^{-n}$ for some $a>1$ (i.e., $-\infty<\alpha_{lim}<0$), then it is impossible to achieve a value of $\Gamma_{RRT}^{\alpha}(k_0)$ closer to one. However, $\Gamma_{RRT}^{\alpha}(k_0)$ will still be bounded away from zero. When $\alpha$ is reduced to zero at a rate faster than $a^{-n}$ for some $a>1$ (say $a^{-n^2}$), then $\Gamma_{RRT}^{\alpha}(k_0)$  converges to zero with increasing $n$. This behaviour of $\Gamma_{RRT}^{\alpha}({k_0})$  have a profound impact on the large sample behaviour of RRT.

Theorem \ref{thm:asymptotic_rrt} is numerically validated in Fig.\ref{fig:asymptotic} where we plot $\Gamma_{RRT}^{\alpha}(k_0)$ for three asymptotic regimes of practical interest, \textit{viz.}, a). $(p,k_0)$ fixed and $n\rightarrow \infty$,  b). $k_0$ fixed while $(p,n)\rightarrow \infty$ and c). $(n,p,k_0)\rightarrow \infty$. In all the three asymptotic regimes, adaptations of $\alpha$ satisfying $\alpha_{lim}= 0$ achieve $\underset{n \rightarrow \infty}{\lim}\Gamma_{RRT}^{\alpha}(k_0)=1$.  These numerical results are in accordance with Theorem \ref{thm:asymptotic_rrt}. 

\subsection{Large sample consistency of RRT}
In this section, we establish the conditions required for the large sample consistency of RRT, i.e., $\underset{n \rightarrow \infty}{\lim}PCS=1$. The main result in this section is Theorem \ref{thm:large_sample} presented below. 
\begin{thm}\label{thm:large_sample} Consider a situation where $n \rightarrow \infty$ such that \\
a) $0\leq k_{lim}=\underset{n \rightarrow \infty}{\lim}k_0/n<1$. \\
b). $\exists M_1>0$ and $n_0<\infty$ satisfying $\|({I}_{n}-{\bf P}_{k_0-1}){\bf x}_{k_0}\|_2^2|\boldsymbol{\beta}_{k_0}|^2/\sigma^2\geq M_1n>0$ for all $n>n_0$. 
Then \\
R1). RRT is large sample consistent provided that the parameter $\alpha$ satisfies $\alpha_{lim}=\underset{n \rightarrow \infty}{\lim}\log(\alpha)/n=0$ and $\underset{n \rightarrow \infty}{\lim}\alpha=0$.\\
R2). For a fixed $0<\alpha\leq 1$, $\underset{n \rightarrow \infty}{\lim}\mathbb{P}_{\mathcal{U}}=0$ and $\underset{n \rightarrow \infty}{\lim}\mathbb{P}_{\mathcal{O}}\leq \alpha$.
\end{thm}
\begin{proof}
Please see Appendix E.
\end{proof}
Theorem \ref{thm:large_sample} implies that with proper adaptations on the parameter $\alpha$, it is possible to achieve a PCS arbitrarily close to one at large sample sizes. We first relate the requirements on $\alpha$ to the probabilities of overestimation and underestimation. 
\begin{remark}  To avoid underestimation at large $n$, i.e., $\underset{n \rightarrow \infty}{\lim}\mathbb{P}_{\mathcal{U}}=0$, it is sufficient that $\alpha_{lim}=0$. By Theorem \ref{thm:asymptotic_rrt}, a fixed value of $\alpha=0.1$ or $\alpha=0.01$ is sufficient for this. The adaptation  $\alpha \rightarrow 0$ as $n \rightarrow \infty$ is necessary only to prevent overestimation.   Further, in addition to the worst case high SNR overestimation probability, the bound $\underset{n \rightarrow \infty}{\lim}\mathbb{P}_{\mathcal{O}}\leq \alpha$ implies that the parameter $\alpha$ in RRT when set independent of $n$ also has the semantic interpretation of worst case large sample overestimation probability. 
\end{remark}
\begin{remark} The only user specified parameter in RRT is $\alpha$. Theorem \ref{thm:large_sample} implies that for all choices of $\alpha$ that satisfies $\alpha_{lim}=0$ and $\underset{n \rightarrow \infty}{\lim}\alpha=0$, RRT will have similar value of PCS at large values of $n$.  Note that the conditions $\alpha_{lim}=0$ and $\underset{n \rightarrow \infty}{\lim}\alpha=0$  are satisfied by a wide range of adaptations like  $\alpha=1/\log(n)$, $\alpha=1/n$  etc.   This points to the insensitivity of RRT to the choice of $\alpha$ as $ n\rightarrow \infty$, i.e., RRT is asymptotically  tuning free.
\end{remark}
We next discuss as corollaries the specific conditions under which the SNR condition in Theorem \ref{thm:large_sample}, i.e., $\exists M_1>0$ and $n_0<\infty$ such that $\|({\bf I}_{n}-{\bf P}_{k_0-1}){\bf x}_{k_0}\|_2^2|\boldsymbol{\beta}_{k_0}|^2/\sigma^2\geq M_1n>0$ for all $n>n_0$ for   hold true. 
\begin{corollary} Let ${\bf X} \in \mathbb{R}^{n \times p}$ with $p<n$ be an orthonormal matrix and $\boldsymbol{\beta}_k=b$ for all $k\leq k_0$ and $-\infty<b<\infty$. Then SNR is given by SNR=$\|{\bf X}\boldsymbol{\beta}\|_2^2/(n\sigma^2)=\|\boldsymbol{\beta}\|_2^2/(n\sigma^2)= k_0b^2/(n\sigma^2)$. Further ${\bf X}$ orthonormal implies that $({\bf I}_n-{\bf P}_{k_0-1}){\bf x}_{k_0}={\bf x}_{k_0}$ and hence $\|({\bf I}_{n}-{\bf P}_{k_0-1}){\bf x}_{k_0}\|_2^2|\boldsymbol{\beta}_{k_0}|^2/\sigma^2=b^2/\sigma^2=n SNR/k_0$. This setting has $M_1=SNR/k_0$ and $n_0=1$. Hence, when $n$ is increased to infinity fixing $k_0$ and SNR constant, RRT is large sample consistent. 
\end{corollary}
\begin{corollary} Following Corollary 1,  consider a situation where $k_0$ is increasing with $n$ and SNR increasing  atleast linearly with  $k_0$ asymptotically, i.e., $SNR/k_0>1$ for some $n\geq n_0$. Then $\|({\bf I}_{n}-{\bf P}_{k_0-1}){\bf x}_{k_0}\|_2^2|\boldsymbol{\beta}_{k_0}|^2/\sigma^2=n SNR/k_0>n$ for $n>n_0$. Here $M_1=1$. Hence if SNR increases atleast linearly with $k_0$,   then RRT is large sample consistent with increasing $k_0$ as long as $k_{lim}=\underset{n \rightarrow \infty}{\lim}k_0/n<1$. When $k_0$ increases and SNR is kept fixed, then  $\|({\bf I}_{n}-{\bf P}_{k_0-1}){\bf x}_{k_0}\|_2^2|\boldsymbol{\beta}_{k_0}|^2/\sigma^2=nSNR/k_0$ increases at the most sub-linearly with $n$ denying the existence of $M_1>0$ and $n_0<\infty$. In that situation, RRT may not be large sample consistent. 
\end{corollary}
\begin{corollary} Next consider the situation ${\bf X} \in \mathbb{R}^{n\times p}$ with $p<n$ and  ${\bf X}_{i,j}\overset{i.i.d}{\sim}\mathcal{N}(0,1/n)$. By Lemma 5 of \cite{cai2011orthogonal}, 
\begin{equation}
\|({\bf I}_{n}-{\bf P}_{k_0-1}){\bf x}_{k_0}\|_2^2\geq \lambda_{min}({\bf X}_{[k_0]}^T{\bf X}_{[k_0]}),
\end{equation}
where $\lambda_{min}({\bf X}_{[k_0]}^T{\bf X}_{[k_0]})$ is the minimum eigenvalue of the  matrix ${\bf X}_{[k_0]}^T{\bf X}_{[k_0]}$. Under the limit $0\leq k_{lim}<1$, it is true that \cite{decoding_candes} 
\begin{equation}
\lambda_{min}({\bf X}_{[k_0]}^T{\bf X}_{[k_0]})\overset{P}{\rightarrow }(1-\sqrt{k_{lim}})^2\ \text{as} \ n \rightarrow \infty.
\end{equation}
Further, the SNR is fixed at $SNR=\dfrac{\mathbb{E}(\|{\bf X}\boldsymbol{\beta}\|_2^2)}{\mathbb{E}(\|{\bf w}\|_2^2)}=\dfrac{\|\boldsymbol{\beta}\|_2^2}{n\sigma^2}=\dfrac{k_0b^2}{n\sigma^2}$. Consequently at large sample sizes, $\|({\bf I}_{n}-{\bf P}_{k_0-1}){\bf x}_{k_0}\|_2^2|\boldsymbol{\beta}_{k_0}|^2/\sigma^2=(1-\sqrt{k_{lim}})^2b^2/\sigma^2=n (1-\sqrt{k_{lim}})^2 SNR/k_0$. It then follows from Corollaries 1-2 that  RRT is  consistent when $n$ increases to $\infty$   such that  \\
a). $k_0$ and SNR are kept fixed. \\
b). $k_0$ increases to $\infty$ and  SNR increases atleast linearly with $k_0$.
\end{corollary}

%
%
%
%
%

\subsection{Comparison between RRT and ITC hyper parameters}
{ In this subsection, we briefly compare the role played by hyper parameter $\alpha$ in RRT and the hyper prameter $v$ in the MOS criteria  of the form $\hat{k_0}=\underset{k=1,2,\dotsc,p}{\arg\min}\|({\bf I}_n-{\bf P}_k){\bf y}\|_2^2+vk$ (like AIC, BIC). It is well known that with increasing values of  $v$, $\mathbb{P}_{\mathcal{O}}$ decreases, whereas, $\mathbb{P}_{\mathcal{U}}$  increases. Exactly similar behaviour is visible in RRT with decreasing values of $\alpha$, i.e., a smaller value of $\alpha$ is qualitatively equivalent to a larger value of penalty parameter $v$.   This observation explains the similarity in the conditions required for large sample consistency of ITC and RRT. Note that to avoid overestimation as $n \rightarrow \infty$, one need $v\rightarrow \infty$ at the rate $v/\log\log(n)\rightarrow \infty$, whereas, to avoid underestimation one would require $v/n \rightarrow 0$, i.e., $v$ should not grow to $\infty$ at a very fast rate \cite{nishii1988maximum}. Once we take into account the fact that smaller $\alpha$ is equivalent to a higher $v$, the conditions that $\alpha \rightarrow 0$ to avoid overestimation and $\log(\alpha)/n \rightarrow 0$ to avoid underestimation are similar to the rules imposed on $v$. Similarly, for fixed values of $v$ and $\alpha$, both ITC and RRT overestimate the model order at high SNR, i.e., $\underset{\sigma^2\rightarrow 0}{\lim}\mathbb{P}_{\mathcal{O}}>0$, whereas, underestimation probability  $\mathbb{P}_{\mathcal{U}}$ satisfies $\underset{\sigma^2\rightarrow 0}{\lim}\mathbb{P}_{\mathcal{U}}=0$\cite{tsp}.}

\section{Numerical Simulations}
In this section, we numerically validate the high SNR and large sample consistency results derived in Section \rom{3}. We also compare the  performance of RRT  with popular MOS techniques. We compare RRT with classical ITC based MOS  like AIC $h(k,\sigma^2_k)=2k$, BIC $h(k,\sigma^2_k)=k\log(n)$ and the recently proposed  PAL in \cite{stoica2013model}. We also consider a recently proposed high SNR consistent (HSC) MOS with penalty $h(k,\sigma^2_k)=\max(k\log(n),2k\log(\frac{1}{\sigma^2_k}))$  \cite{tsp}. By Lemma 1, this technique is HSC as long as $n>2k_0$ and this condition is true in all our experiments. A technique to design penalty functions based on the high SNR behaviour of $\mathbb{P}_{\mathcal{O}}$ is proposed in \cite{designITC}.  This technique is also implemented (called ``Design" in figures) with  desired error levels $0.1$ and $0.01$ shown in brackets. Simulation results for other popular algorithms like EEF, NMDL, g-MDL etc. are not included because of space constraints. However, we have observed that the relative performance comparisons between RRT and algorithms like PAL, Design, BIC etc. also hold true  for NMDL, EEF etc.    

The entries of matrix ${\bf X}$ are sampled independently from $\mathcal{N}(0,1)$ and the columns are later normalised to have unit $l_2$ norm. We consider two models for $\boldsymbol{\beta}$, (1)  model 1 has $\boldsymbol{\beta}_k=\pm 1$ for all $k\leq k_0$ (i.e., signal component of $\boldsymbol{\beta}$ given by  $\boldsymbol{\beta}_{[k_0]}$ is not sparse) and 2) model 2  has $\boldsymbol{\beta}_k=\pm 1$ only  for few entries between $k=1$ and $k_0$ (i.e., $\boldsymbol{\beta}_{[k_0]}$ is  sparse). The non-zero locations will be reported alongside the figures.   Model 2 is  typical of auto regressive (AR) model order selection where the maximum lag (i.e., true order of AR process $k_0$) can be very high, however, the generator polynomial has only  few non-zero coefficients. Likewise, in sparse channel estimation\cite{raghavendra2005improving,tomasoni2013efficient}, it is likely that the length of channel impulse response (i.e., $k_0$) is high.  However, the CIR contains only few non-zero coefficients. Model 2 represents this scenario too.  All the results presented in this section are obtained after $10^4$ iterations. 
\subsection{Validating Theorem \ref{thm:highSNR} and Theorem \ref{thm:large_sample}}
\begin{figure*}
\begin{multicols}{2}
\includegraphics[width=\linewidth]{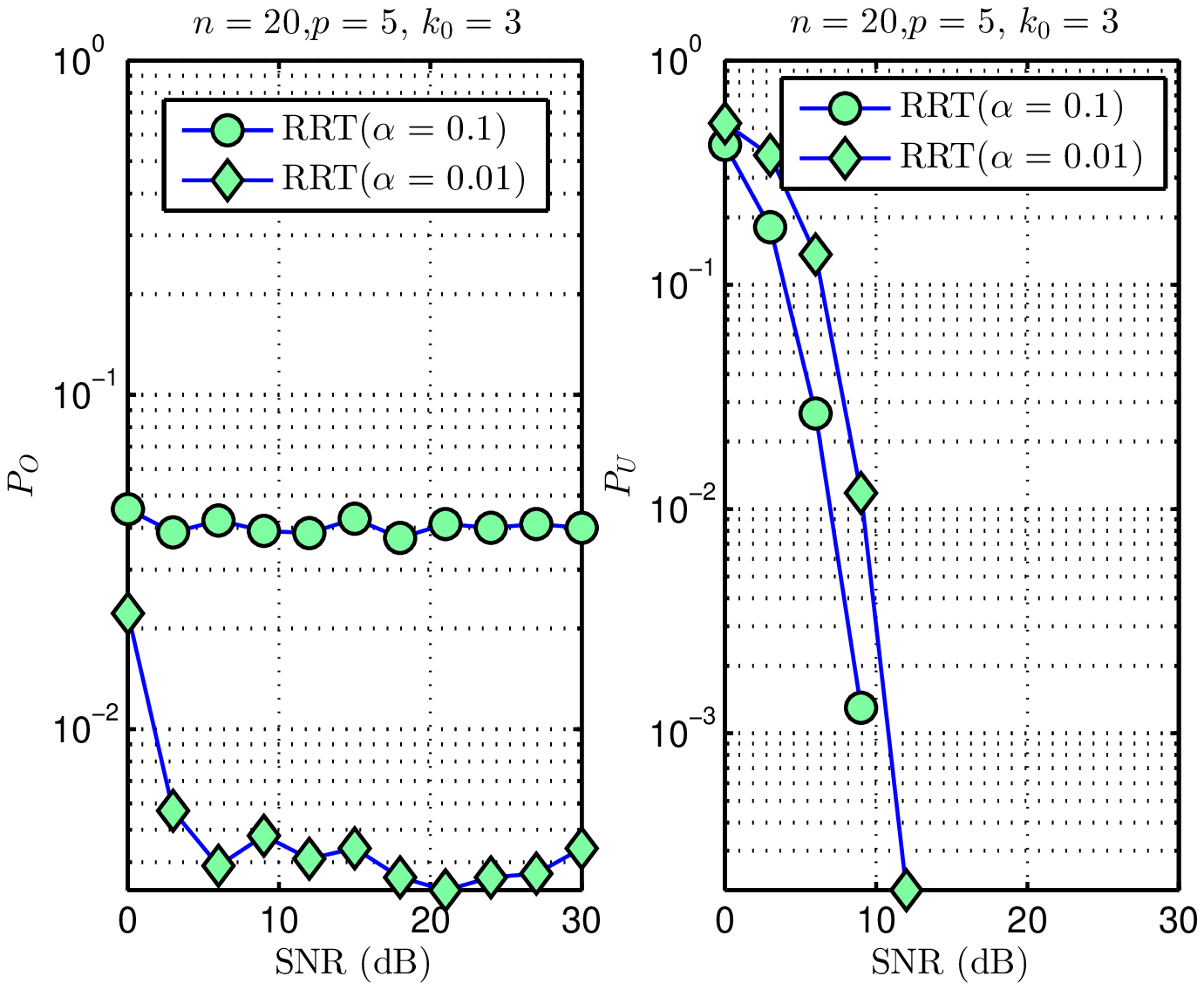}
\caption*{a)Verification of Theorem \ref{thm:highSNR}.$n=20$, $p=10$ and $k_0=3$.}
\includegraphics[width=\linewidth]{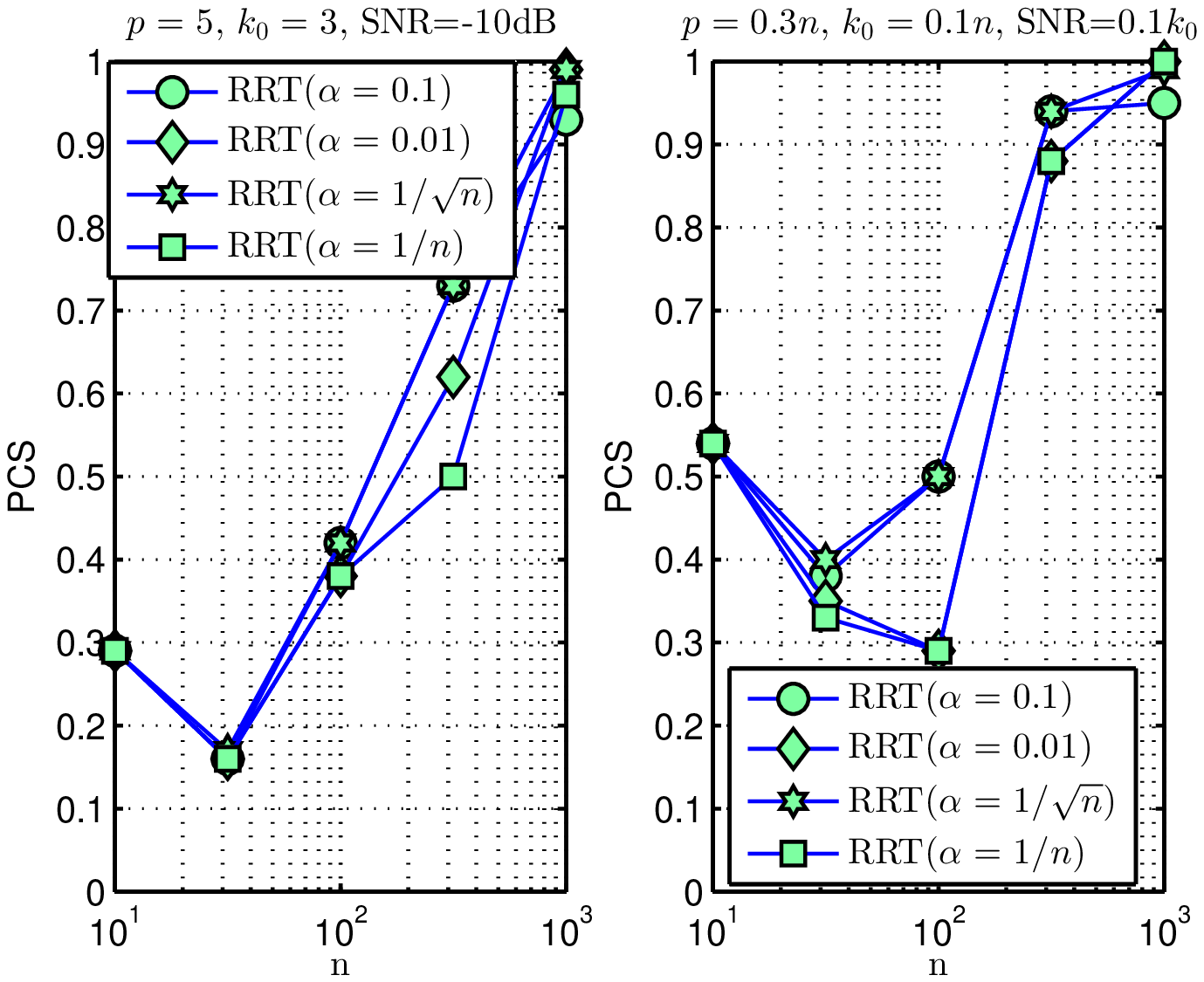}
\caption*{b). Verification of Theorem \ref{thm:large_sample}. $p=5$,$k_0=3$ and SNR=-10dB (left).$p=0.3n$,$k_0=0.1n$ and SNR=0.1$k_0$ (right).}
\end{multicols}
\squeezeup
\caption{Verification of Theorem \ref{thm:highSNR} and Theorem \ref{thm:large_sample}. $\boldsymbol{\beta}_k=\pm 1$ for all $k\leq k_0$.}
\label{fig:verification}
\squeezeup
\end{figure*}

In this section, we numerically validate the high SNR and large sample results presented in Theorem \ref{thm:highSNR} and Theorem \ref{thm:large_sample} of Section \rom{3}.  
Fig.\ref{fig:verification}.a) presents the variations in $\mathbb{P}_{\mathcal{O}}$ and $\mathbb{P}_{\mathcal{U}}$ with increasing SNR. From the L.H.S of Fig.\ref{fig:verification}.a), one can see that $\mathbb{P}_{\mathcal{O}}$ floors at $\approx 10^{-1.5}$ from zero dB SNR onwards when $\alpha=0.1$ and $\approx 10^{-2.5}$ from 3dB SNR onwards when $\alpha=0.01$. These  evaluated values of $\mathbb{P}_{\mathcal{O}}$ satisfy the bound $\underset{\sigma^2 \rightarrow 0}{\lim}\mathbb{P}_{\mathcal{O}}\leq \alpha$ predicted by Theorem \ref{thm:highSNR}. In fact, the bound $\mathbb{P}_{\mathcal{O}}\leq \alpha$ hold true even at a low SNR of 3dB. Likewise, as one can see from the R.H.S of Fig.\ref{fig:verification}.a), $\mathbb{P}_{\mathcal{U}}$ decreases with increasing SNR. This is also in accordance with the limit $\underset{\sigma^2 \rightarrow 0}{\lim}\mathbb{P}_{\mathcal{U}}=0$ predicted by Theorem \ref{thm:highSNR}. Note that restricting overestimation probability to smaller values in MOS problems will always  leads to an increase in finite SNR  underestimation probability for any MOS criterion. This explains  the increase in $\mathbb{P}_{\mathcal{U}}$ for $\alpha=0.01$ at finite SNR compared to  $\alpha=0.1$. 

Fig.\ref{fig:verification}.b) presents the variations in PCS with increasing sample size $n$. Among the four choices of $\alpha$ considered, only $\alpha=1/\sqrt{n}$ and $\alpha=1/n$ can lead to large sample consistency according to Theorem \ref{thm:large_sample}.  We consider two regimes of interest. Regime 1 depicted in the L.H.S of Fig.\ref{fig:verification}.b) deals with the situation where $n$ increases to $\infty$ keeping $p$, $k_0$ and  SNR  fixed. As one can see from Fig.\ref{fig:verification}.b),  PCS for all values of $\alpha$ increases to one with increasing $n$. However,   $PCS$ for $\alpha=0.01$ and $\alpha=0.1$ floor near one satisfying the bounds $\underset{n \rightarrow \infty}{\lim}PCS\geq 1-\alpha$ in R2) of Theorem \ref{thm:large_sample}, whereas, PCS for $\alpha=1/\sqrt{n}$ and $\alpha=1/n$ converges to one satisfying R1) of Theorem \ref{thm:large_sample}. Regime 2 deals with a situation where all $n$, $p$ and $k_0$ increases to $\infty$ with SNR increasing linearly with $k_0$. As one can see from Fig.\ref{fig:verification}.b), PCS for $\alpha=1/n$ and $\alpha=1/\sqrt{n}$ converge to one, whereas, $PCS$ for $\alpha=0.01$ and $\alpha=0.1$ floor near one satisfying the bounds $\underset{n \rightarrow \infty}{\lim}PCS\geq 1-\alpha$. These results validate Theorem \ref{thm:large_sample} and its' corollaries. 
\begin{figure*}
\begin{multicols}{2}

    \includegraphics[width=1\linewidth]{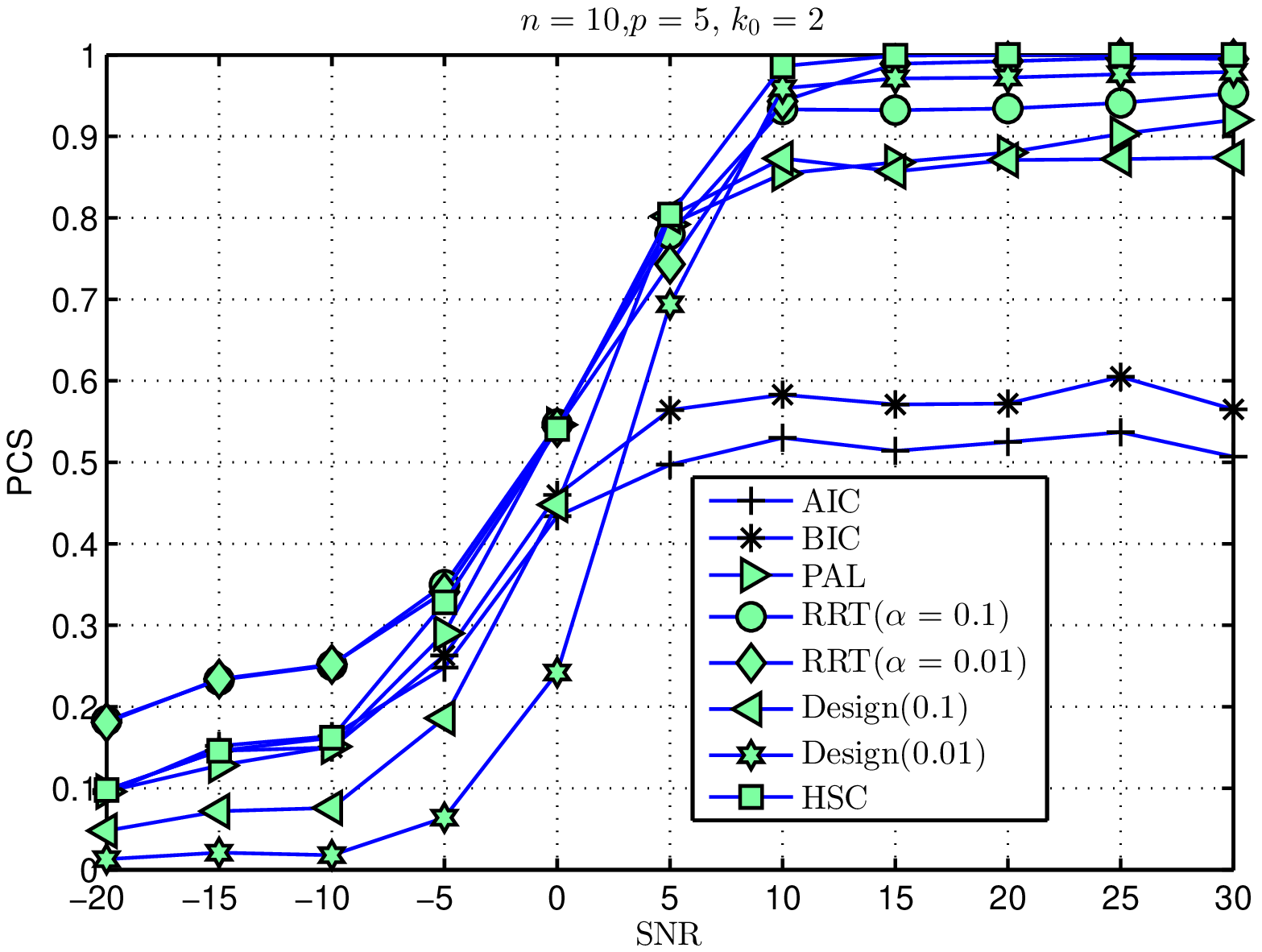} 
    \caption*{a). $n=10$, $p=5$ and $k_0=2$.}
    
    \includegraphics[width=1\linewidth]{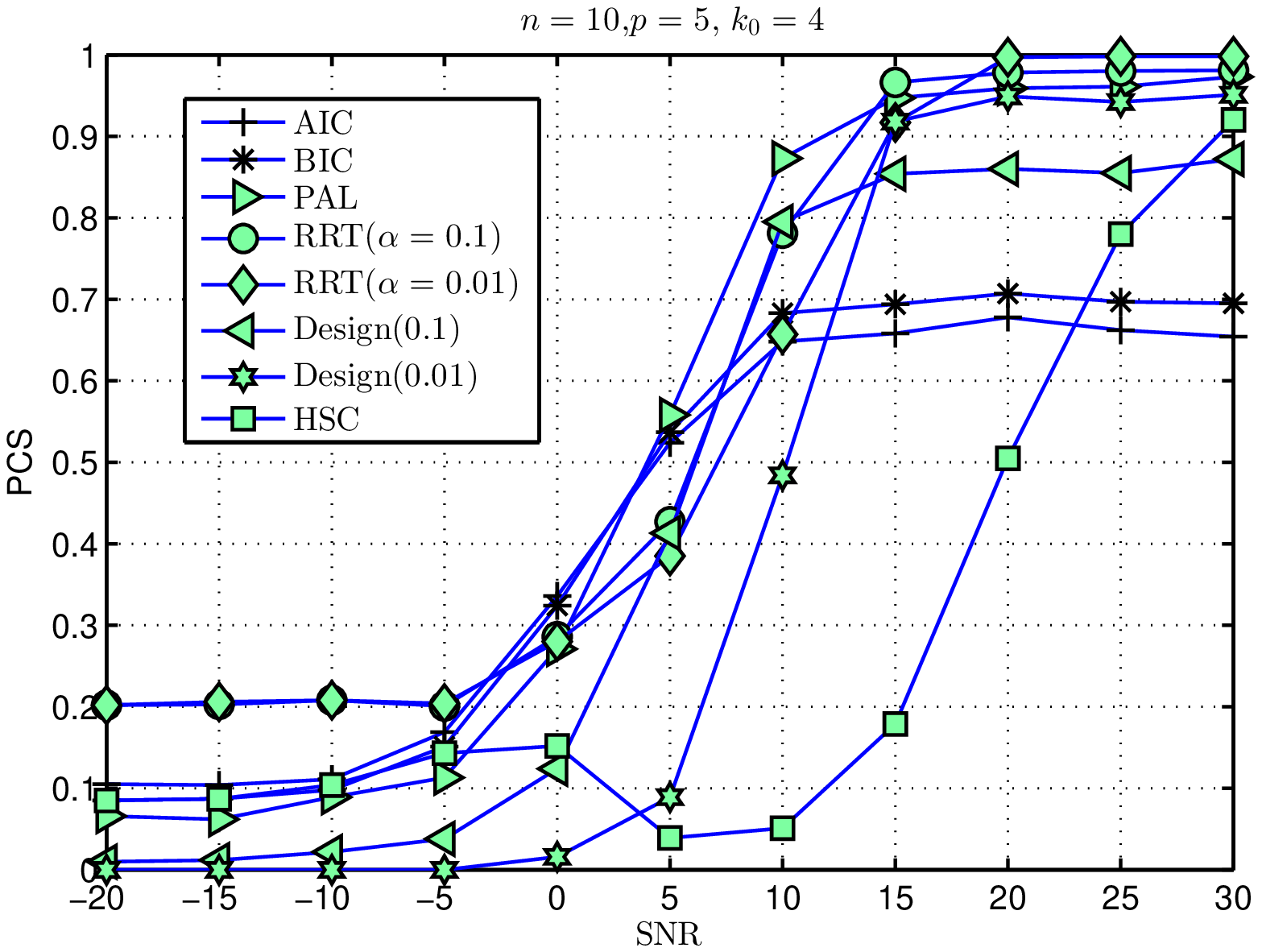} 
    \caption*{b). $n=10$, $p=5$ and $k_0=4$.}
    \end{multicols}
   
   \begin{multicols}{2}

    \includegraphics[width=1\linewidth]{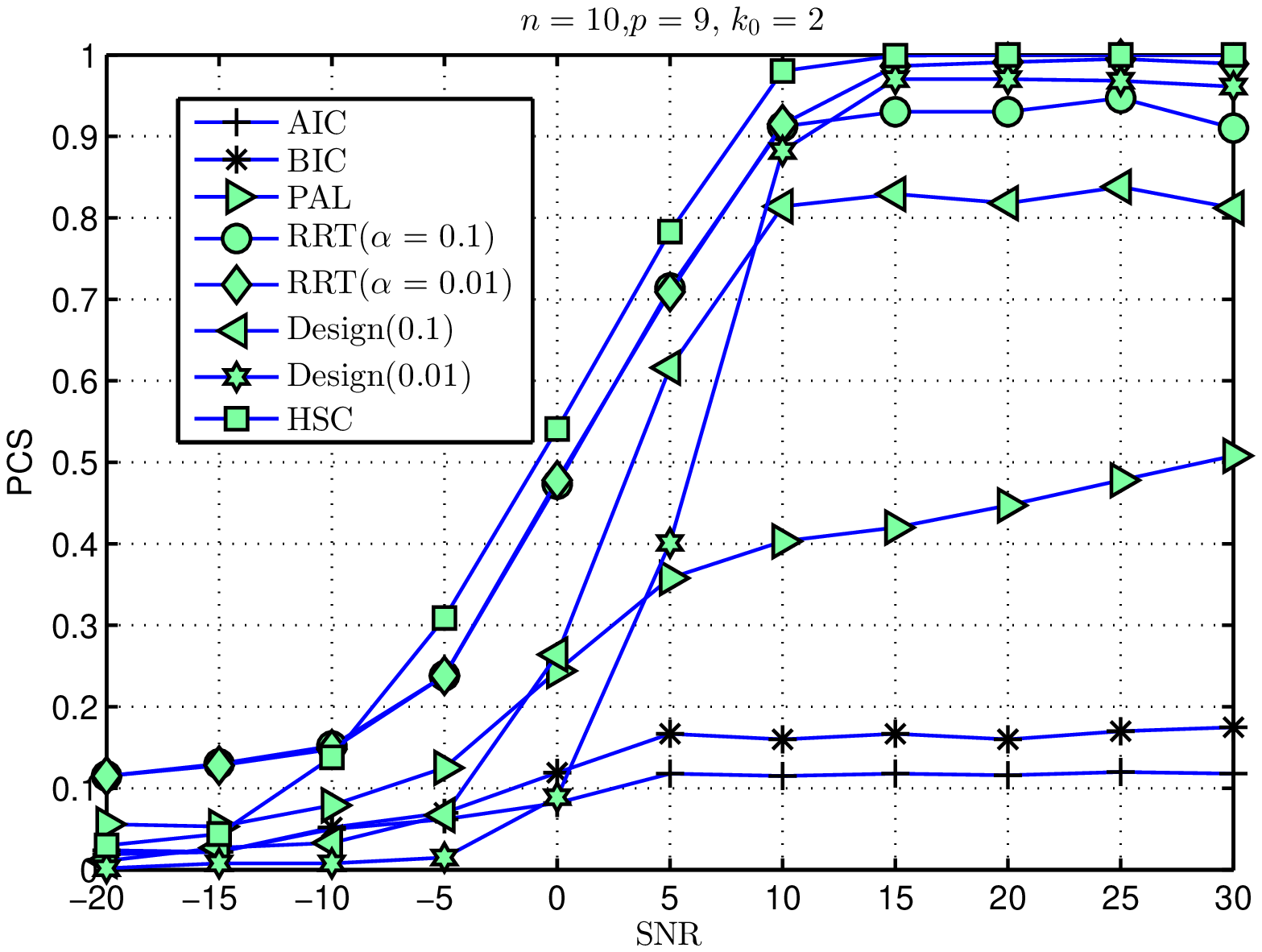} 
    \caption*{c). $n=10$, $p=9$ and $k_0=2$.}
    
    \includegraphics[width=1\linewidth]{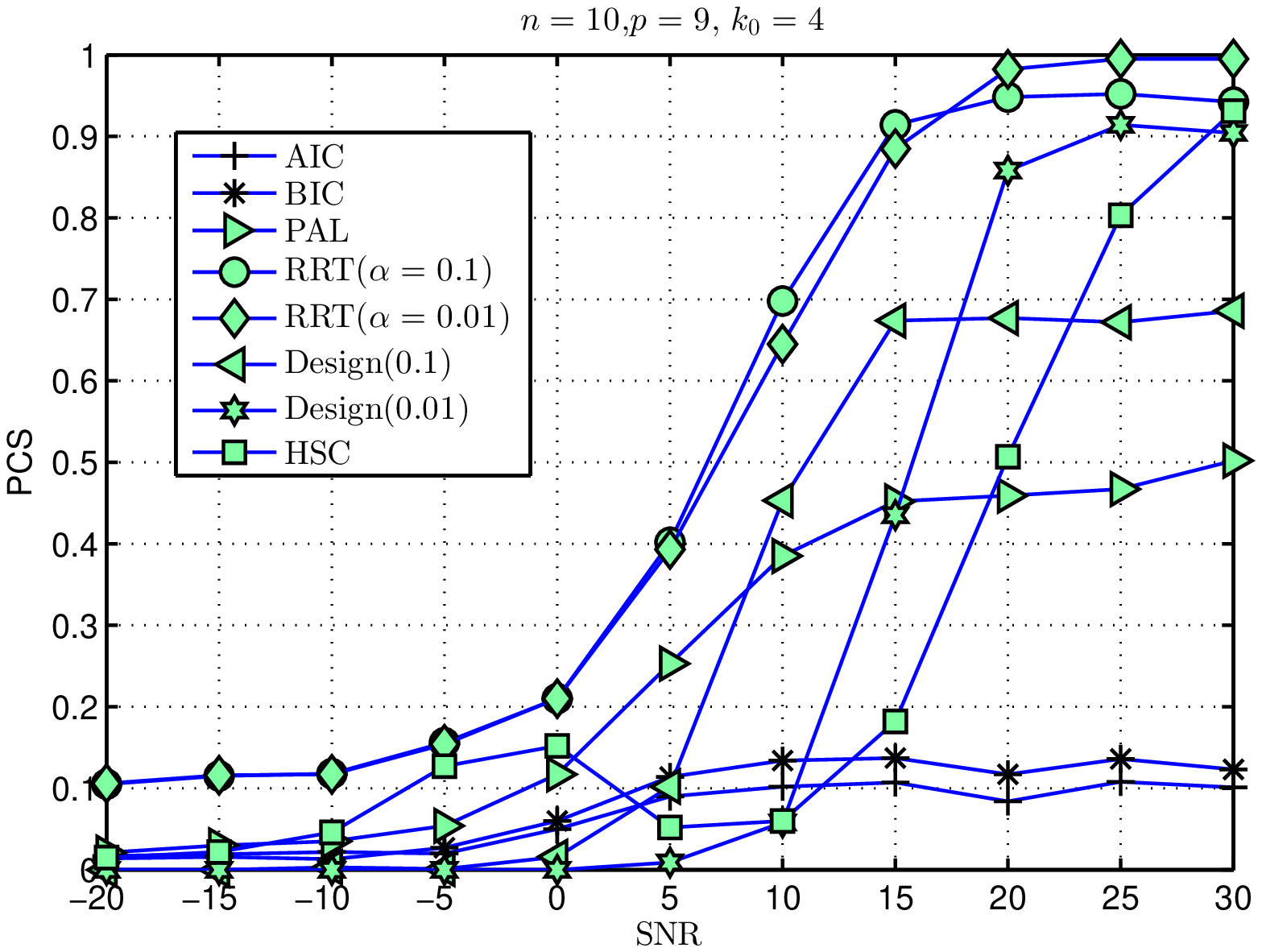} 
    \caption*{d).  $n=10$, $p=9$ and $k_0=4$.}
    \end{multicols} 
    \squeezeup
   \caption{Small sample performance: $\boldsymbol{\beta}_k=\pm 1$ for all $k\leq k_0$.}
   \label{fig:small_sample}
   \squeezeup
\end{figure*}

\subsection{Experiment 1: PCS when both  $n$ and $p$ are small }

We first compare the PCS performance of MOS techniques when sample size $n$ is very small in absolute terms. Fig.\ref{fig:small_sample}.a) and Fig.\ref{fig:small_sample}.b) illustrate a situation where $p$ is much smaller than $n$, i.e., $p=n/2$. When $k_0=2$, one can see that RRT with $\alpha=0.1$ and $\alpha=0.01$ outperform other algorithms at very low SNR. In the medium SNR regime, HSC,PAL, Design($0.1$) and RRT with both $\alpha=0.1$ and $\alpha=0.01$ have similar performances. At high SNR, the best performance is delivered by RRT with $\alpha=0.01$, Design($0.01$) and HSC. The PCS of PAL appears to floor below one. The reason for this is the slow growth of penalty function in PAL with increasing SNR\cite{tsp} which causes overestimation. PCS of AIC and BIC are much inferior compared to other   MOS techniques. 
When $k_0$ is increased to $k_0=4$, performance of HSC  deteriorates significantly. The high SNR performance of PAL, AIC, MDL, RRT etc. improve when $k_0=4$. {This can be reasoned as follows. At high SNR, the error in MOS criteria like AIC, BIC, PAL etc. is overwhelmingly due to overestimation. Please note that when $k_0$ is increased keeping $p$ constant, the probability of overestimation decreases. This explains the improvement in PCS  with increasing $k_0$ for MOS criteria  like  PAL, AIC, BIC etc. which have a tendency to overestimate at high SNR. HSC incorporates a SNR adaptation to the BIC penalty to decrease its' $\mathbb{P}_{\mathcal{O}}$. Note that any attempt to decrease overestimation probability will result in an increase in underestimation probability in  the low to moderate SNR. Since, with increasing $k_0$, the importance of overestimation decreases and underestimation increases, the SNR adaptation intended for avoiding overestimation will result in more underestimation in the low to moderate SNR regime. The explains the deteriorating performance of HSC with increasing $k_0$.}
\begin{figure*}[htb]
\begin{multicols}{2}
\includegraphics[width=1\linewidth]{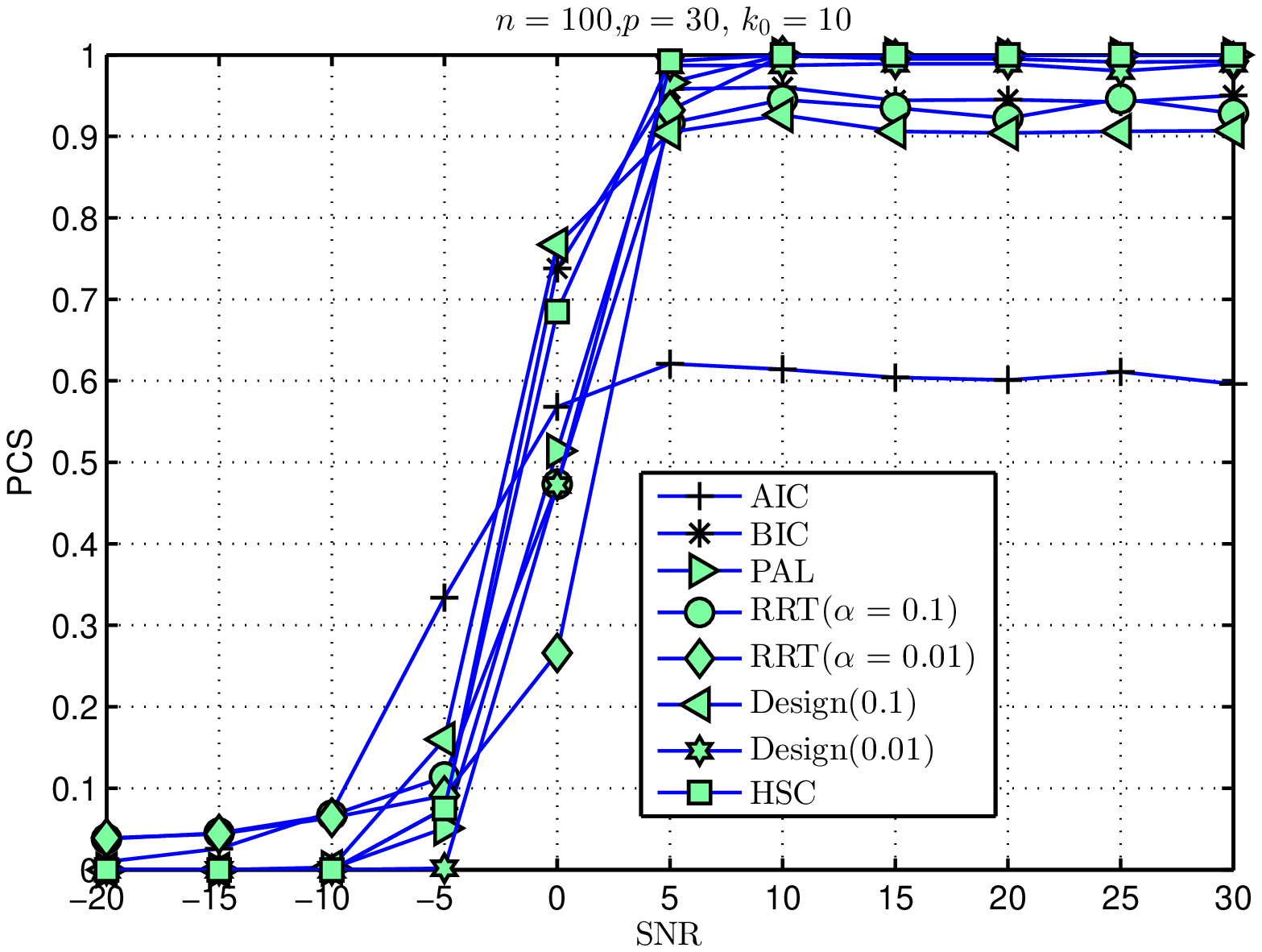} 
    \caption*{a). $n=100$, $p=30$ and $k_0=10$.}
    
    \includegraphics[width=1\linewidth]{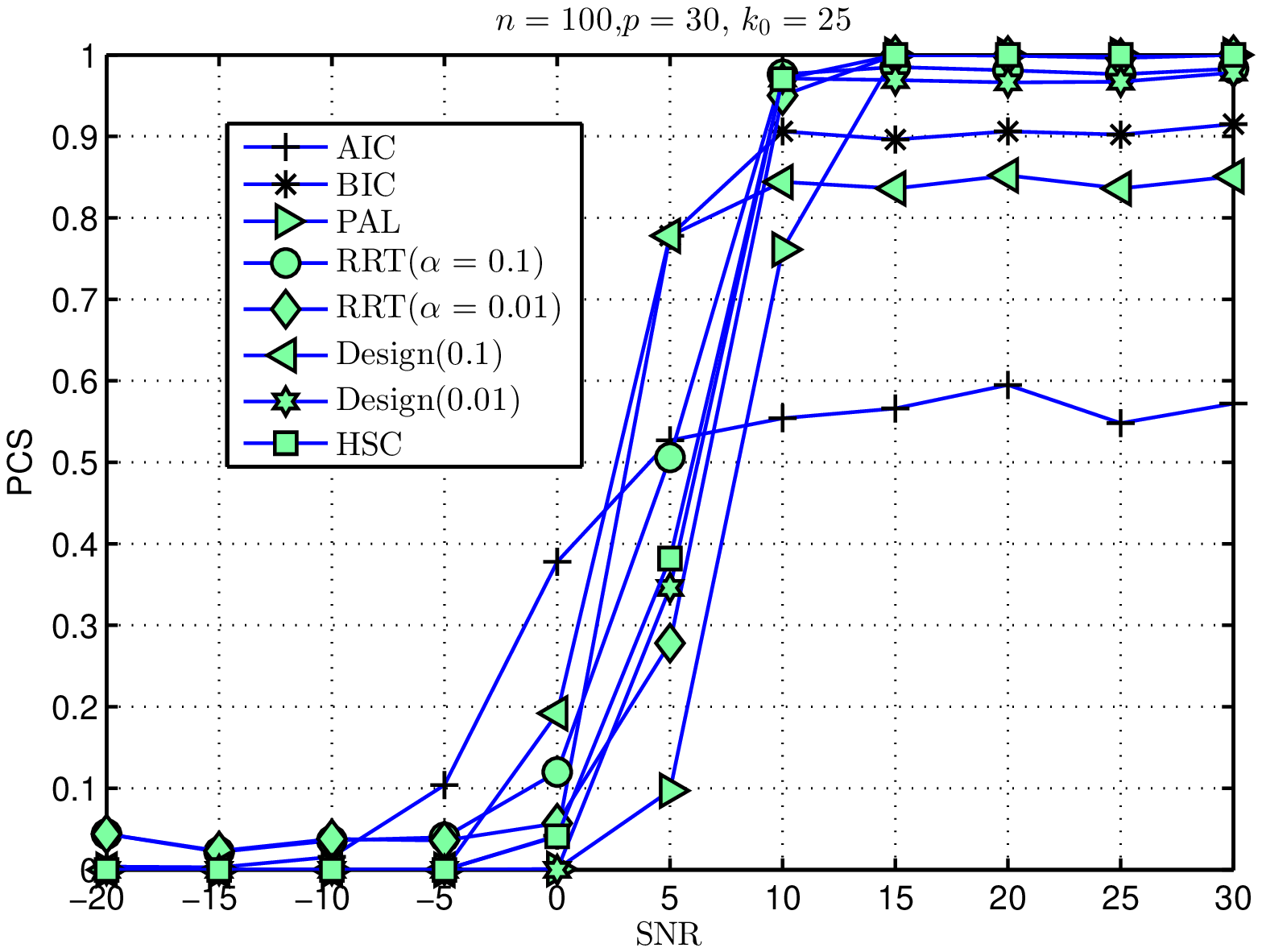} 
    \caption*{b). $n=100$, $p=30$ and $k_0=25$.}
    \end{multicols}
   
   \begin{multicols}{2}

    \includegraphics[width=1\linewidth]{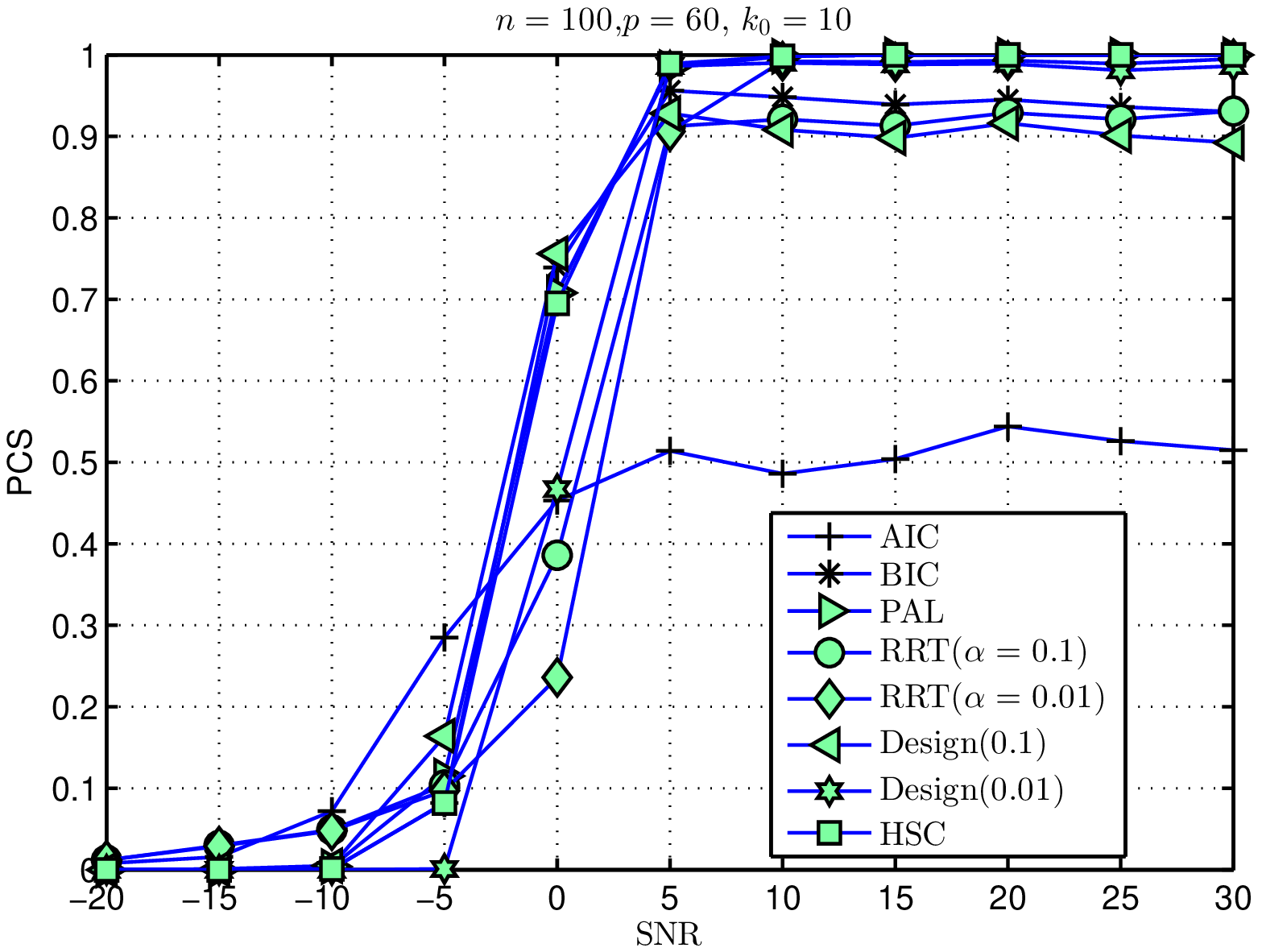} 
    \caption*{c). $n=100$, $p=60$ and $k_0=10$.}
    
    \includegraphics[width=1\linewidth]{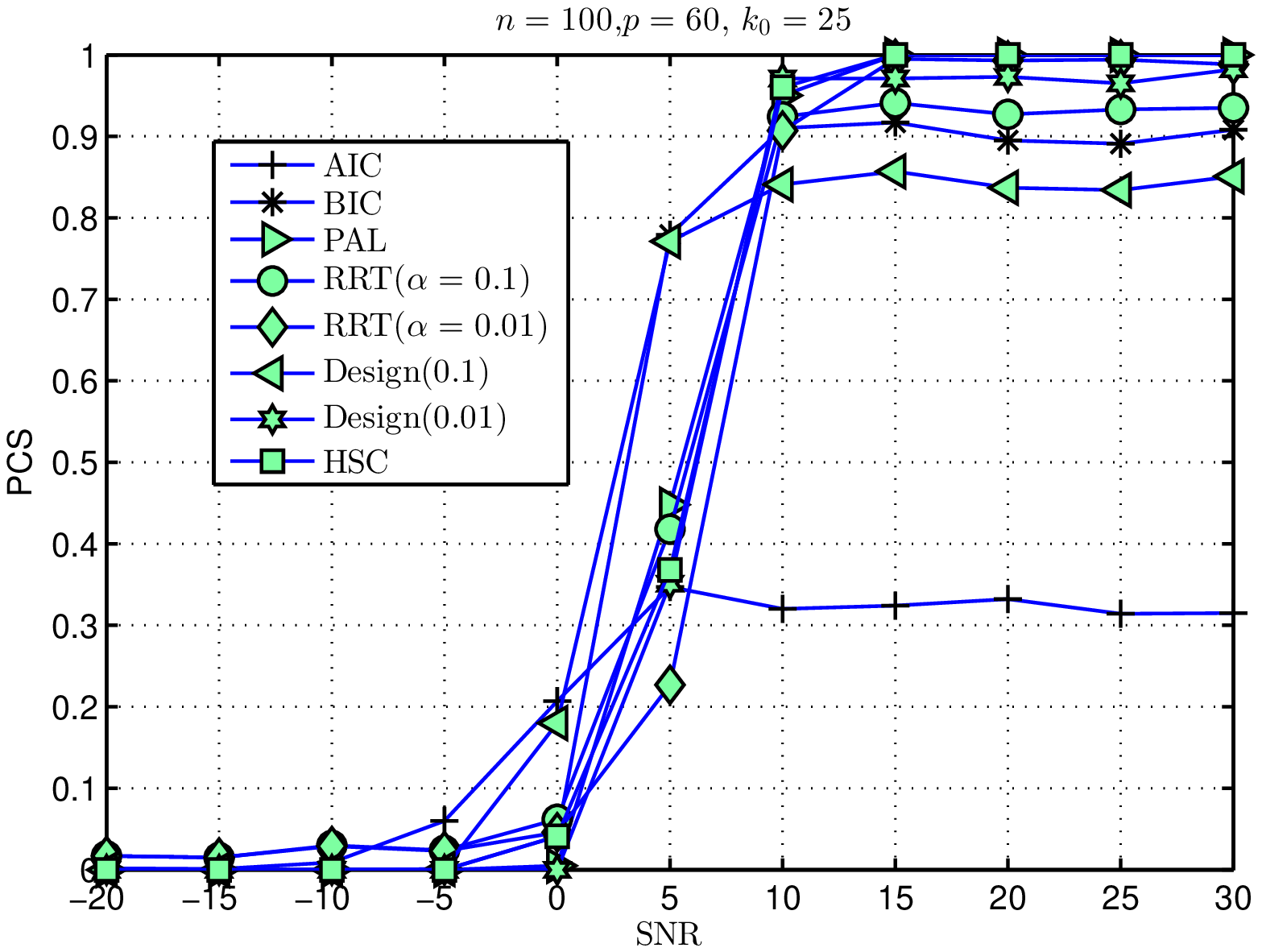} 
    \caption*{d).  $n=100$, $p=60$ and $k_0=25$.}
    \end{multicols} 
    \squeezeup
   \caption{$\boldsymbol{\beta}_k=\pm 1$ for $k=1,2...\dotsc,k_0$. $\boldsymbol{\beta}_{[k_0]}$ is long and dense.}
   \label{fig:nonsparse}
   \squeezeup
\end{figure*}

Next we compare the performance of MOS techniques when $p$ and $n$ are nearly the same. As one can see from Fig.\ref{fig:small_sample}.c) and Fig.\ref{fig:small_sample}.d),  performances of PAL, AIC and BIC are much worse in this case than with $p=5$. Again this is because of the fact that $\mathbb{P}_{\mathcal{O}}$ increases when $p$ is increased while keeping $k_0$ fixed.  When $k_0=2$, HSC achieves the best overall performance. However, when $k_0=4$, the performance of HSC is remarkably poor in the low to moderate high SNR regime. This general trend of HSC performing badly with increasing $k_0$ is observed in many other simulations too. When $k_0=4$ and $p=9$, RRT with both $\alpha=0.1$ and $\alpha=0.01$ outperform all other algorithms by a significant margin.    From Fig.\ref{fig:small_sample}, it is clear that  no single algorithm  outperformed all other algorithms in all the four scenarios. However, RRT delivered the best performance in atleast  one scenario, whereas, it delivered near best performance in all the three other scenarios. Further, RRT with $\alpha=0.1$ outperformed  Design($0.1$) and RRT with $\alpha=0.01$ outperformed Design($0.01$) in all the four experiments.  This is significant considering the fact both these schemes guarantee same value of high SNR error probability.
\subsection{Experiment 2: PCS when $n$ is large and  SNR is varying }

\begin{figure*}[htb]
\begin{multicols}{2}

    \includegraphics[width=1\linewidth]{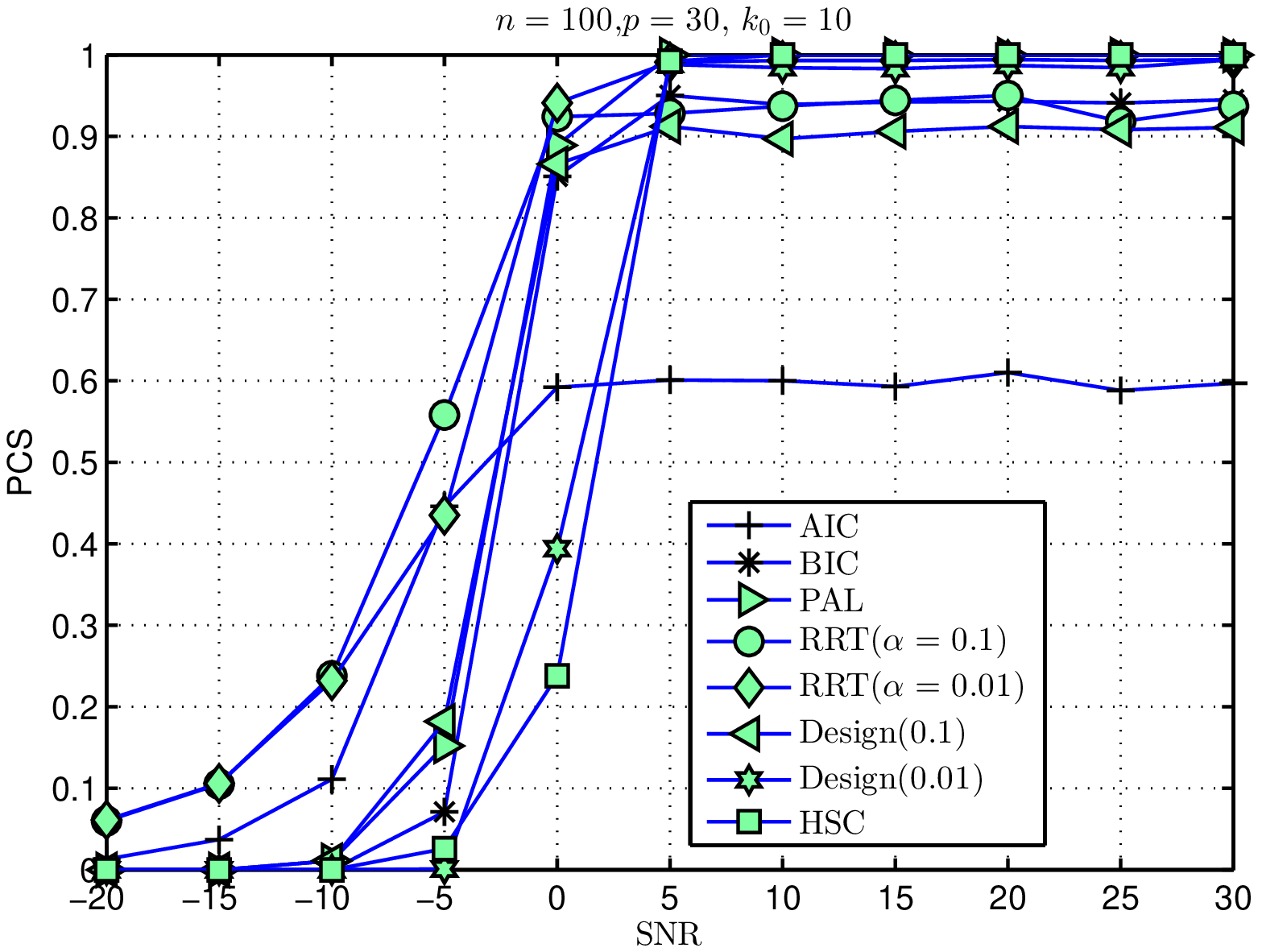} 
    \caption*{a). $n=100$, $p=30$ and $k_0=10$.}
    
    \includegraphics[width=1\linewidth]{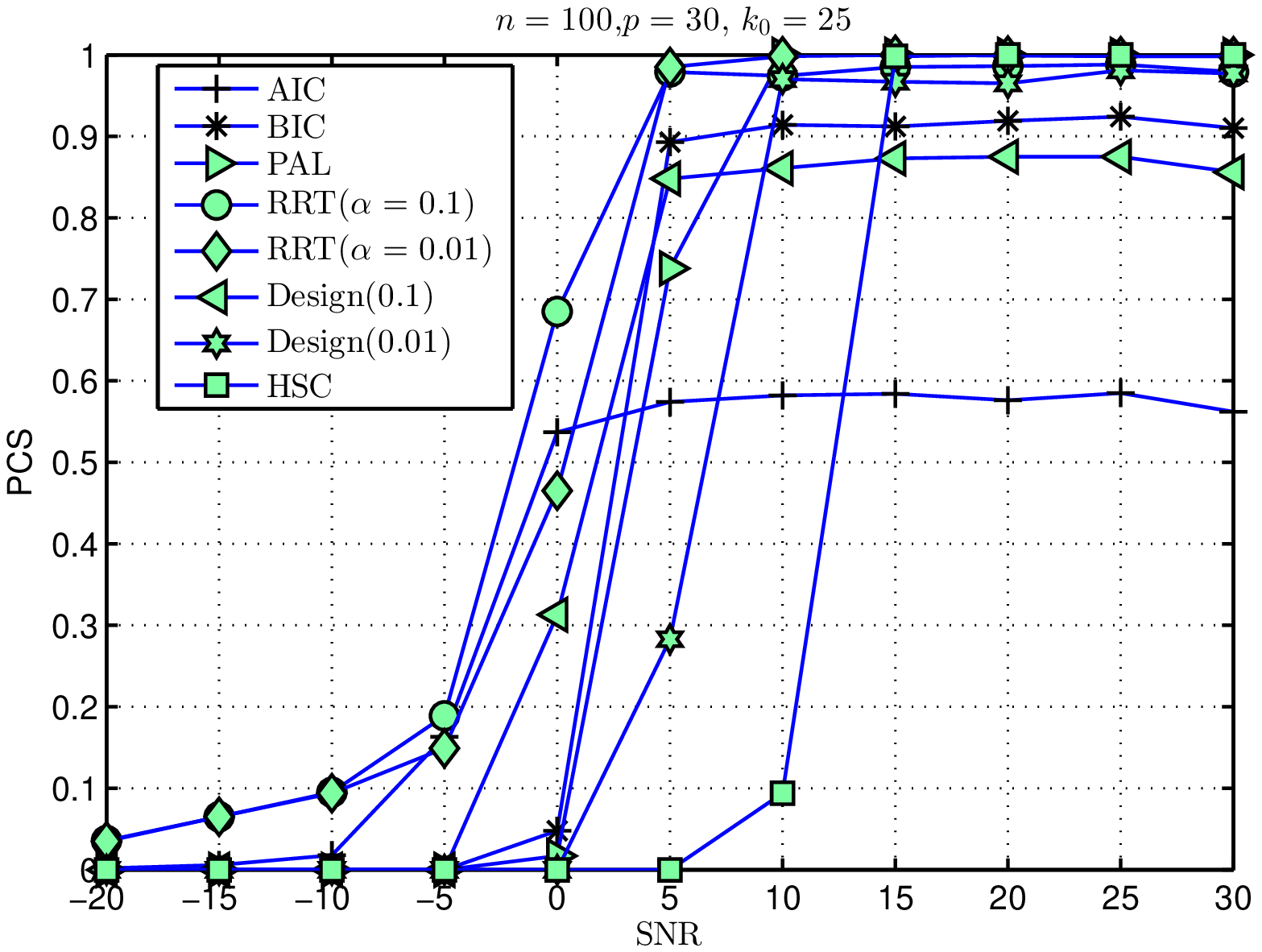} 
    \caption*{b). $n=100$, $p=30$ and $k_0=25$.}
    \end{multicols}
   
   \begin{multicols}{2}

    \includegraphics[width=1\linewidth]{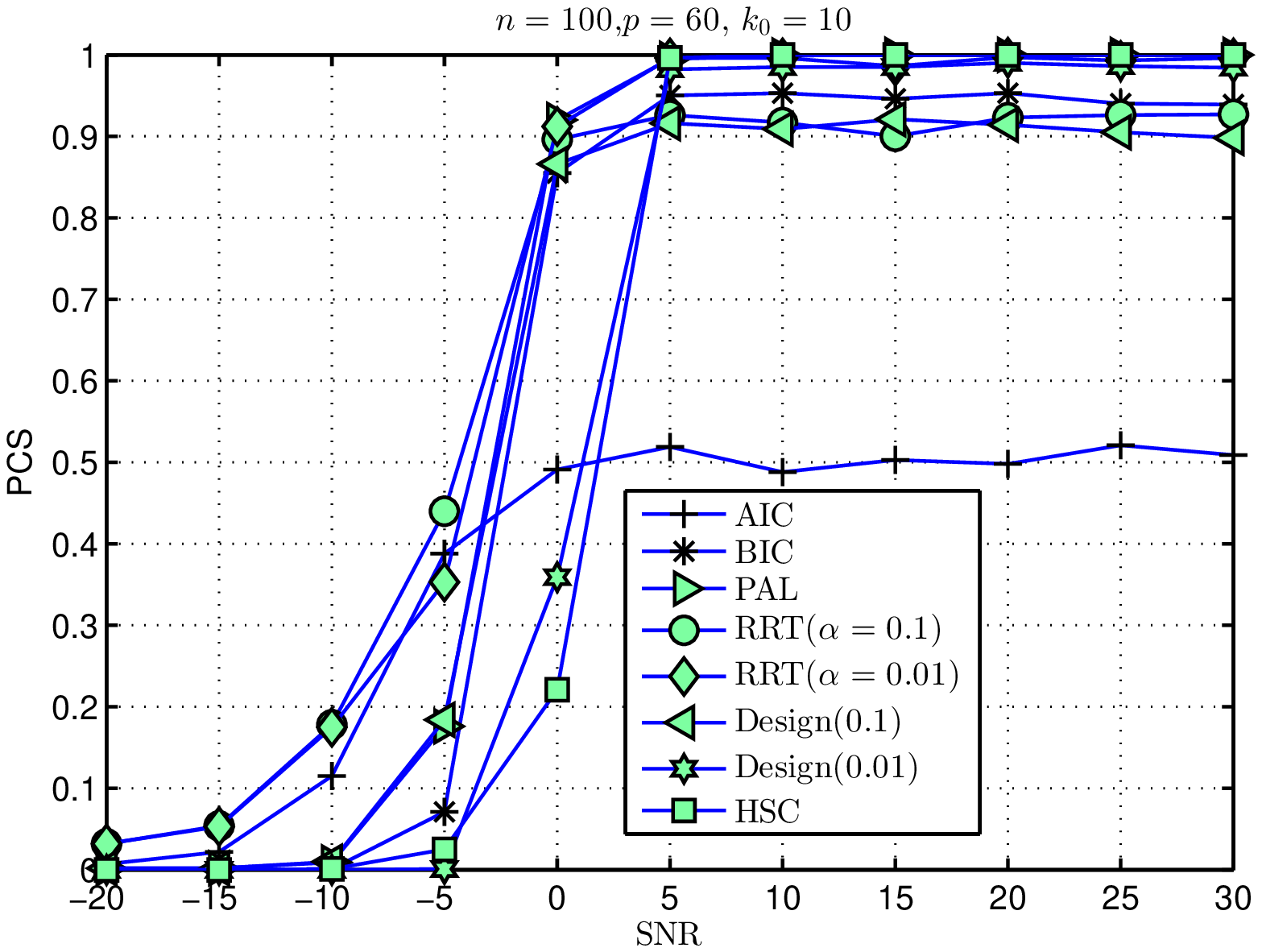} 
    \caption*{c). $n=100$, $p=60$ and $k_0=10$.}
    
    \includegraphics[width=1\linewidth]{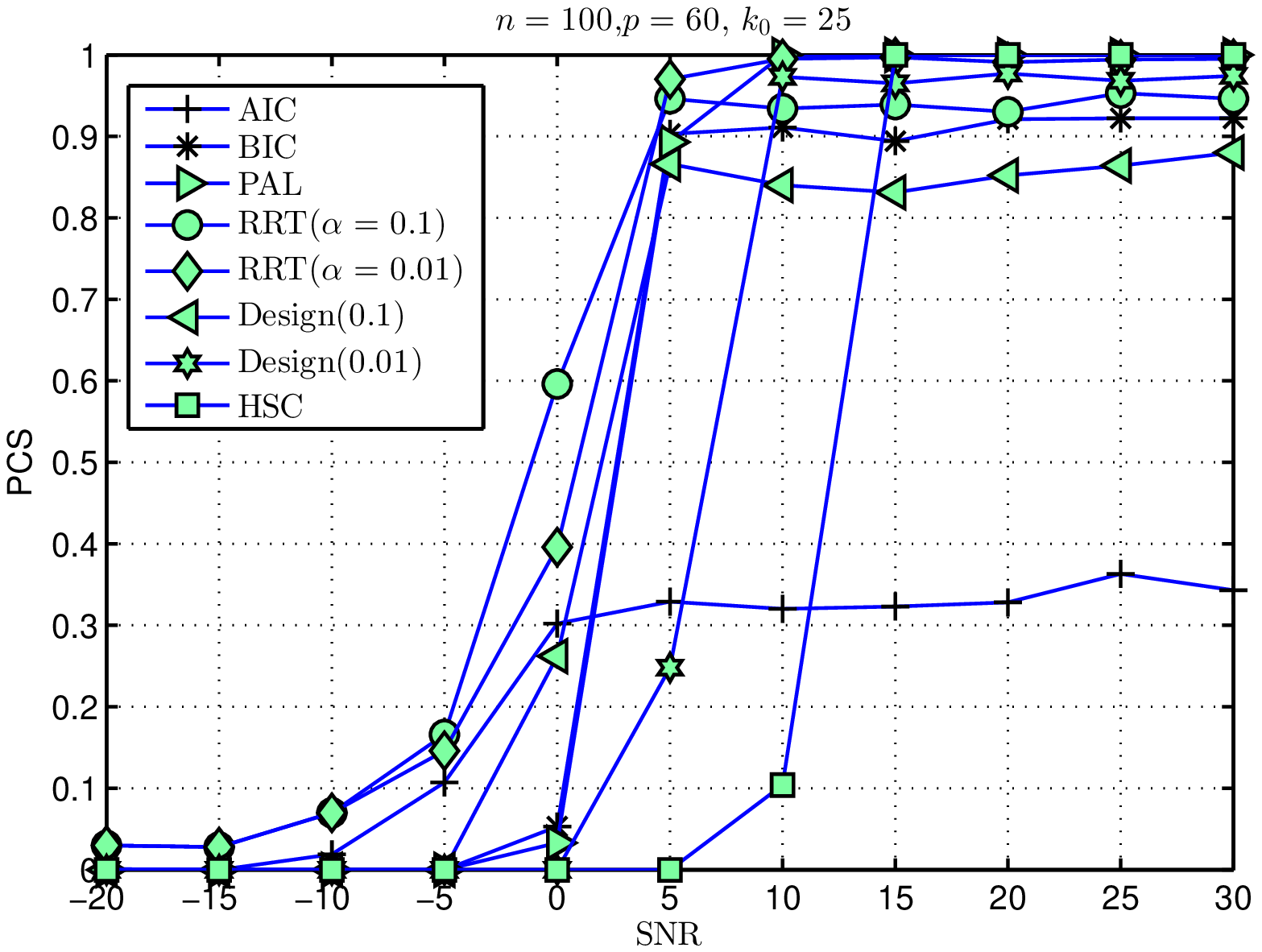} 
    \caption*{d).  $n=100$, $p=60$ and $k_0=25$.}
    \end{multicols} 
    \squeezeup
   \caption{$\boldsymbol{\beta}_k=\pm 1$ for $k=1,6...\dotsc,$ and $k=k_0$. $\boldsymbol{\beta}_{[k_0]}$ is long but sparse.}
   \label{fig:sparse}
   \squeezeup
\end{figure*}

Next we consider the performance of algorithms with increasing SNR when the problem dimensions $(n,p,k_0)$ are moderately large.  From the PCS figures for Model 1 given in Fig.\ref{fig:nonsparse}, it is clear that the performance of algorithms like AIC, BIC, PAL etc. have improved tremendously compared to the case when $n$ was set at $n=10$. From the four scenarios considered in Fig.\ref{fig:nonsparse}, it is difficult to pick a single winner. However, apart from Fig.\ref{fig:nonsparse}.c, in all the other situations RRT with $\alpha=0.1$ performed closer to most of the other algorithms for all values of SNR. Unlike the previous case, Design($0.1$)  does outperforms RRT with $\alpha=0.1$ many often. 

Next we consider the performance of algorithms when $\boldsymbol{\beta}$ is of Model 2, i.e., sparse. Unlike the case of Model 1, RRT with $\alpha=0.1$ is a clear winner throughout the low to high SNR in all the four experiments considered in Fig.\ref{fig:sparse}. In fact this trend of RRT performance improving with increasing  sparsity of $\boldsymbol{\beta}_{[k_0]}$ and a corresponding deterioration in the performance of ITC based MOS techniques was observed in  a large number of experiments conducted. 
 \subsection{Experiment 3: PCS of algorithms with increasing $n$}
 \begin{figure*}[htb]
\begin{multicols}{2}

    \includegraphics[width=1\linewidth]{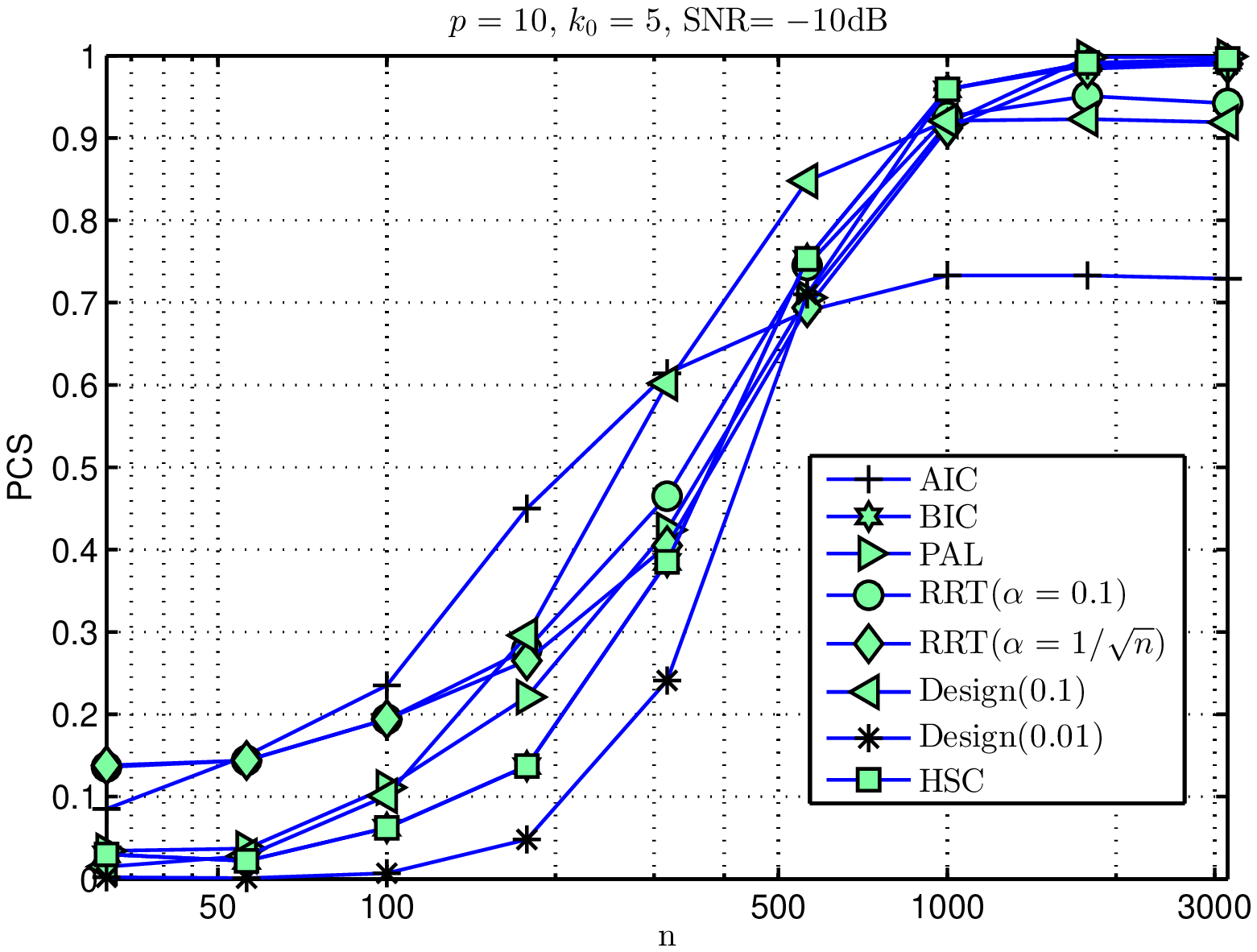} 
    \caption*{a).  $p=10$, $k_0=5$ and SNR=$-10$dB.}
    
    \includegraphics[width=1\linewidth]{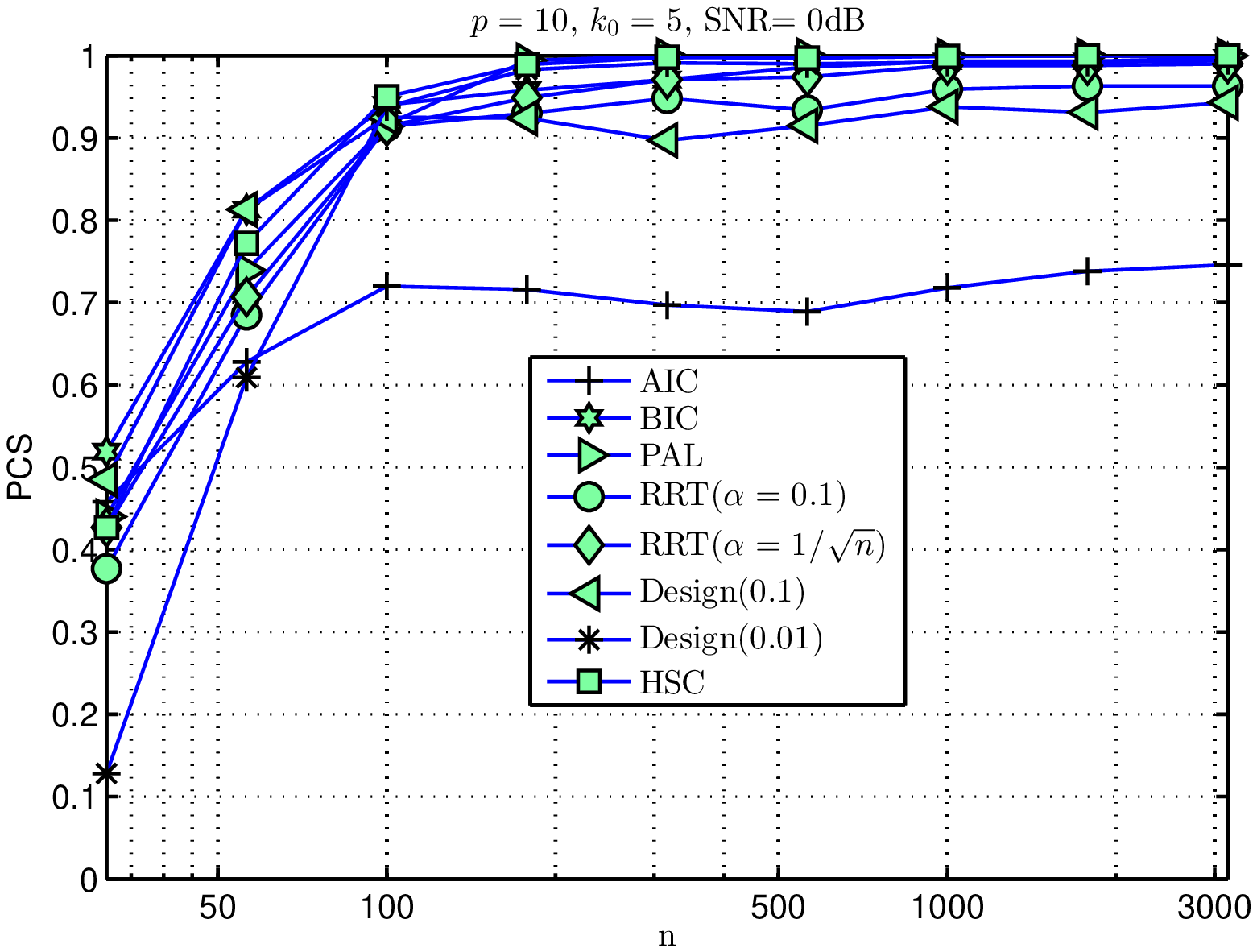} 
    \caption*{b). $p=10$, $k_0=5$ and SNR=$0$dB.}
    \end{multicols}
   
   \begin{multicols}{2}

    \includegraphics[width=1\linewidth]{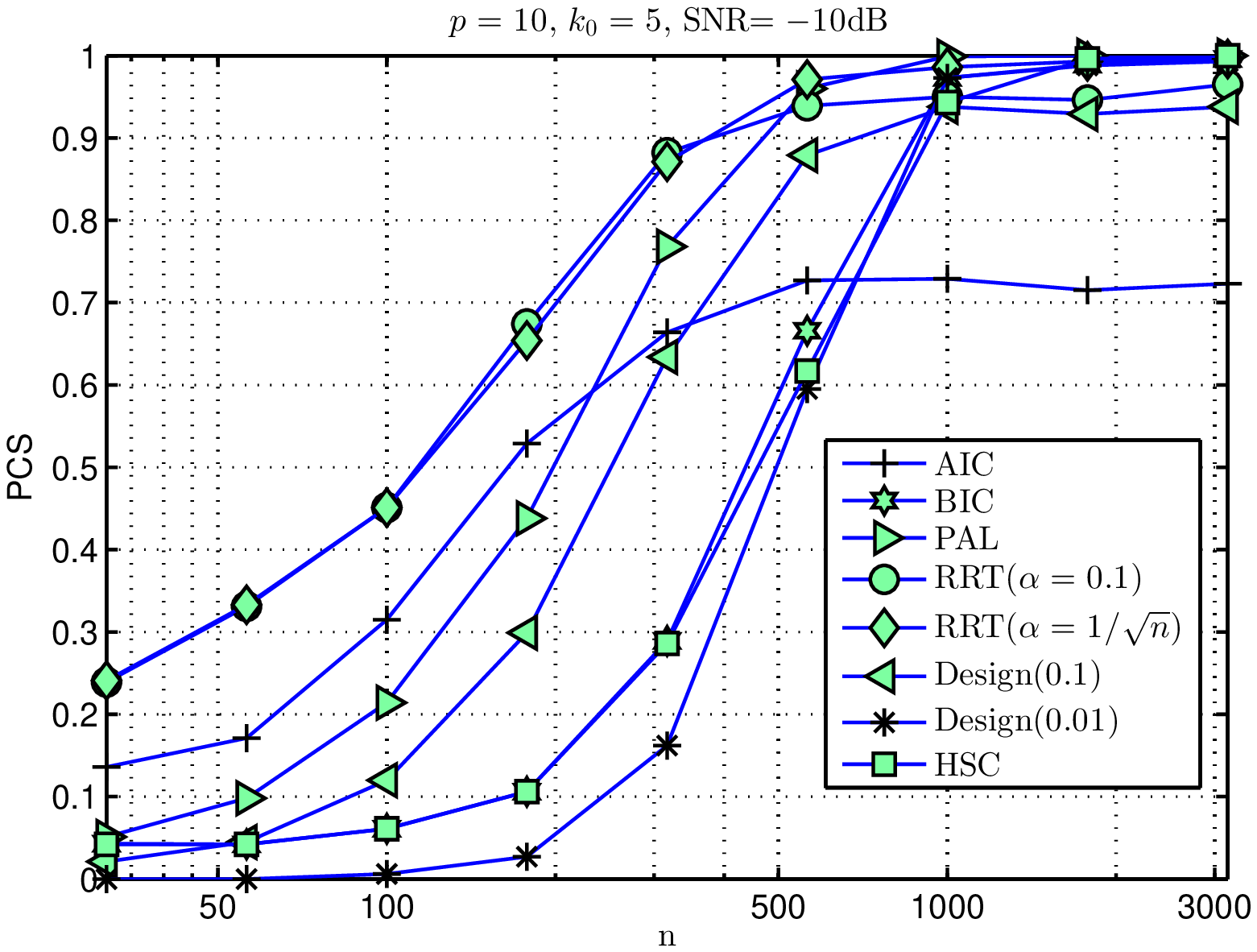} 
    \caption*{c).  $p=10$, $k_0=5$ and SNR$=-10$dB.}
    
    \includegraphics[width=1\linewidth]{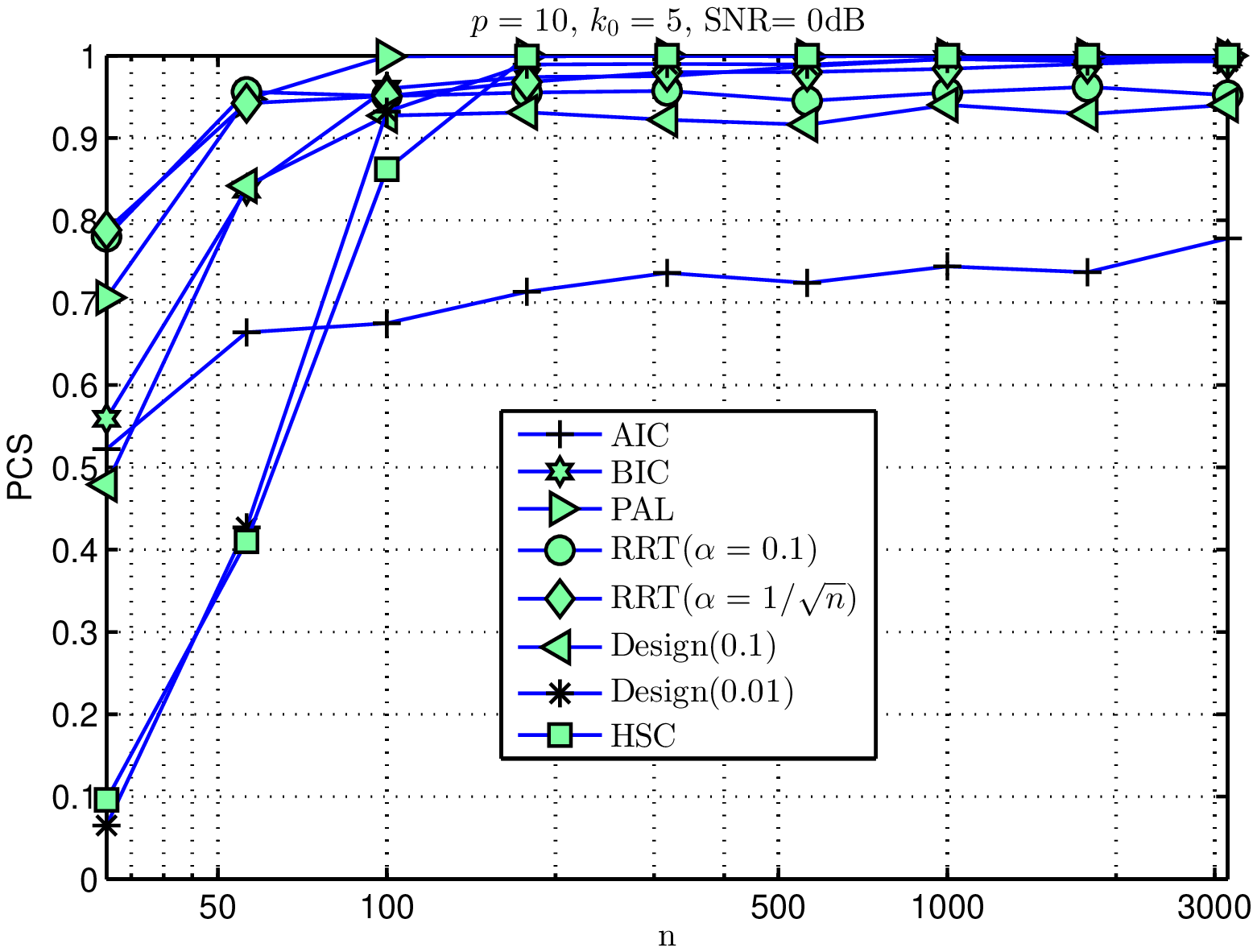} 
    \caption*{d). $p=10$, $k_0=5$ and SNR$=0$dB.}
    \end{multicols} 
    \squeezeup
   \caption{Large sample performance:  $\boldsymbol{\beta}_k=\pm1,\forall k\leq k_0$ for a) and b).  $\boldsymbol{\beta}_k=\pm1$ for  $k=1$ and $k= k_0=5$ for c)  and d). }
   \label{fig:largesample}
   \squeezeup
\end{figure*}
 
 We depict in  Fig.\ref{fig:largesample} the performance of MOS criteria when the sample size $n$ is increasing while keeping  SNR, $k_0$ and $p$ fixed. When  $\boldsymbol{\beta}$ is of Model 1, one can see  from  Fig.\ref{fig:largesample}.a) and b)  that RRT performs as good as most of the other algorithms under consideration. However,  when $\boldsymbol{\beta}$ is of Model 2, it is clear from Fig.\ref{fig:largesample}.c) and d) that RRT with $\alpha=0.1$ and $\alpha=0.01$ clearly outperform all the other algorithms. 
 
 To summarize, RRT has definitive performance advantages over many existing MOS techniques when the sample size $n$ is very small. When the sample size $n$ is large and $\boldsymbol{\beta}_{[k_0]}$ is dense, RRT did  not exhibit any significant performance advantages. Indeed, the observed performance of RRT in the low to moderately high SNR is inferior compared to the best performing MOS criteria like Design($0.1$). However, when the vector $\boldsymbol{\beta}_{[k_0]}$ is sparse, RRT clearly outperformed all the other MOS criteria under consideration. 
 \subsection{Choice of $\alpha$ in RRT }
The performance of RRT depends crucially on the choice of $\alpha$. From the 18 experiments presented in this section and many other experiments not shown in this article, we found out that $\alpha=0.1$ delivered the best overall PCS performance in the low to moderately high SNR regime. Indeed, this choice is purely empirical. However, even with this choice, one can guarantee a value of $\mathbb{P}_{\mathcal{O}}$ less than $10\%$ throughout the operating SNR regimes. When the SNR is very high, a situation not so common in practical applications, one can set $\alpha$ to smaller values like $\alpha=0.01$. Likewise, when the sample size $n$ is very large, one can set $\alpha=1/\sqrt{n}$ which was also found to deliver a very good performance. Finding a completely data dependent choice of $\alpha$ in RRT is of tremendous operational importance  and will be part of the  future research.  

%
%
%
%
%
%
%

\section{Conclusions and directions of future research}
This article proposes a novel MOS criterion based on the behaviour of residual norm ratios. The proposed technique is philosophically different from the widely used ITC based MOS  techniques. This article also provides high SNR and large sample performance guarantees for RRT. In particular, the large sample consistency of RRT is established.  Numerical simulations too demonstrate a highly competitive  performance of the proposed technique over many widely used MOS techniques. Extending the operational philosophy of RRT to non linear model order selection problems like source number enumeration and developing completely data dependent choices for hyper parameter $\alpha$ are two possible avenues for extending the RRT technique proposed in this article.   
\section*{Appendix A:Proof of Lemma \ref{lemma:basic_distributions} \cite{yanai2011projection},\cite{tsp}}
Since the model order is $k_0$, $\boldsymbol{\beta}_{k}=0$ for $k> k_0$. Consequently, the signal component in ${\bf y}$, i.e., ${\bf X}\boldsymbol{\beta}$ is equal to ${\bf X}_{[k_0]}\boldsymbol{\beta}_{[k_0]}$. 
Hence, ${\bf X}\boldsymbol{\beta}\in span({\bf X}_{[k_0]})$. This along with the full rank assumption on ${\bf X}$   implies that $({\bf I}_n-{\bf P}_k){\bf X}\boldsymbol{\beta}=({\bf I}_n-{\bf P}_k){\bf X}_{[k_0]/[k]}\boldsymbol{\beta}_{[k_0]/[k]}\neq {\bf 0}_n$ for $k<k_0$, whereas, $({\bf I}_n-{\bf P}_k){\bf X}\boldsymbol{\beta}={\bf 0}_n$ for $k\geq k_0$. Consequently, $({\bf I}_n-{\bf P}_k){\bf y}=({\bf I}_n-{\bf P}_k){\bf X}_{[k_0]/[k]}\boldsymbol{\beta}_{[k_0]/[k]}+({\bf I}_n-{\bf P}_k){\bf w}$ for $k<k_0$ and  $({\bf I}_n-{\bf P}_k){\bf y}=({\bf I}_n-{\bf P}_k){\bf w}$ for $k\geq k_0$. The distribution of norm of Gaussian vectors is given in the following Lemma. 
\begin{lemma}\label{lemma:chi2}\cite{yanai2011projection}
Let ${\bf x} \sim \mathcal{N}({\bf u},\sigma^2{\bf I}_n)$ and ${\bf P}\in \mathbb{R}^{n \times n}$ be any projection matrix of rank $j$. Then, \\
a). ${\bf P}{\bf x}\sim \mathcal{N}({\bf P}{\bf u},\sigma^2{\bf P} )$.\\
b). $\|{\bf P}{\bf x}\|_2^2/\sigma^2\sim \chi^2_j(\frac{\|{\bf P}{\bf u}\|_2^2}{\sigma^2})$ if ${\bf Pu}\neq {\bf 0}_n$.\\
c). $\|{\bf P}{\bf x}\|_2^2/\sigma^2\sim \chi^2_j$ if ${\bf Pu}= {\bf 0}_n$.
\end{lemma}
Since ${\bf P}_k$ is a projection matrix  of rank $k$, $({\bf I}_n-{\bf P}_k)$ is a projection matrix of rank $n-k$. Hence, by   Lemma \ref{lemma:chi2}  
\begin{equation}
\|{\bf r}^k\|_2^2/\sigma^2=\chi^2_{n-k}\left(\dfrac{\|({\bf I}_n-{\bf P}_k){\bf X}_{[k_0]/[k]}\boldsymbol{\beta}_{[k_0]/[k]}\|_2^2}{\sigma^2}\right)\text{\ for\ } k<k_0
\end{equation}
\begin{equation}\label{residual1}
\text{and} \ \|{\bf r}^k\|_2^2/\sigma^2=\chi^2_{n-k}\text{\ for\ } k\geq k_0
\end{equation}
Since  ${\bf P}_k{\bf P}_{k-1}={\bf P}_{k-1}$, $({\bf I}_n-{\bf P}_k)({\bf P}_k-{\bf P}_{k-1})={\bf O}_n$. This implies that $\|{\bf r}^{k-1}\|_2^2=\|({\bf I}_n-{\bf P}_k+{\bf P}_k-{\bf P}_{k-1}){\bf y}\|_2^2=\|{\bf r}^{k}\|_2^2+\|({\bf P}_k-{\bf P}_{k-1}){\bf y}\|_2^2$. Note that $({\bf P}_k-{\bf P}_{k-1})$ is a projection matrix of rank one projecting onto the subspace $span({\bf X}_{[k]})\cap span({\bf X}_{[k-1]})^{\perp}$, i.e., $\{{\bf v} \in \mathbb{R}^n: {\bf v} \in span({\bf X}_{[k]}) \& {\bf v} \notin span({\bf X}_{[k-1]})\}$. This implies that $({\bf P}_k-{\bf P}_{k-1}){\bf X}\boldsymbol{\beta}=({\bf P}_k-{\bf P}_{k-1}){\bf X}_{[k_0]}\boldsymbol{\beta}_{[k_0]}\neq {\bf 0}_n$  for $k\leq k_0$, whereas, $({\bf P}_k-{\bf P}_{k-1}){\bf X}\boldsymbol{\beta}={\bf 0}_n$ for $k>k_0$. This implies
\begin{equation}\label{residual12}
\dfrac{\|({\bf P}_k-{\bf P}_{k-1}){\bf y}\|_2^2}{\sigma^2}\sim \chi^2_1\left(\dfrac{\|({\bf P}_k-{\bf P}_{k-1}){\bf X}_{[k_0]}\boldsymbol{\beta}_{[k_0]}\|_2^2}{\sigma^2}\right) \text{\ for \ } k\leq k_0
\end{equation}
\begin{equation}\label{residual2}
\text{and} \ \|({\bf P}_k-{\bf P}_{k-1}){\bf y}\|_2^2/\sigma^2\sim \chi^2_1 \text{\ for\ } k> k_0
\end{equation}

The  orthogonality of matrices ${\bf I}_n-{\bf P}_k$ and ${\bf P}_k-{\bf P}_{k-1}$ implies that the vectors 
${\bf r}_k=({\bf I}_n-{\bf P}_k){\bf y}$ and $({\bf P}_k-{\bf P}_{k-1}){\bf y}$ are uncorrelated. Since  these vectors are Gaussian, they and their norms ($\|{\bf r}^{k}\|_2^2$, $\|({\bf P}_k-{\bf P}_{k-1}){\bf y}\|_2^2$)  are independent too.  
\begin{lemma}\label{lemma:Beta_def}\cite{ravishanker2001first}
Let $Z_1\sim \chi^2_{k_1}$ and $Z_2\sim \chi^2_{k_2}$ be two independent $\chi^2$ R.Vs. Then the ratio $\dfrac{Z_1}{Z_1+Z_2}\sim \mathbb{B}(\dfrac{k_1}{2},\dfrac{k_2}{2})$
\end{lemma}

Using (\ref{residual1}) and (\ref{residual2}) along with the independence of $({\bf I}_n-{\bf P}_k){\bf y}$ and $({\bf P}_k-{\bf P}_{k-1}){\bf y}$ in Lemma \ref{lemma:Beta_def} gives the following distributional result. 
\begin{equation}\label{beta_prelim}
RR(k)=\dfrac{\|{\bf r}^k\|_2^2/\sigma^2}{\|{\bf r}^{k-1}\|_2^2/\sigma^2}\sim \dfrac{\chi^2_{n-k}}{\chi^2_{n-k}+\chi^2_1}\sim \mathbb{B}(\dfrac{n-k}{2},\dfrac{1}{2})
\end{equation}
for $k>k_0$ and all $\sigma^2>0$. This is a) of Lemma \ref{lemma:basic_distributions}.

From (\ref{residual1}) and (\ref{residual12}),  $RR(k_0)=\frac{{ Z}_1}{Z_1+Z_2}$, where $Z_1=\|({\bf I}_n-{\bf P}_{k_0}){\bf w}\|_2^2\sim \sigma^2\chi^2_{n-k_0}$ and $Z_2=\|({\bf P}_{k_0}-{\bf P}_{k_0-1}){\bf y}\|_2^2\sim\sigma^2 \chi^2_{1}\left(\dfrac{\|({\bf P}_{k_0}-{\bf P}_{k_0-1}){\bf X}_{[k_0]}\boldsymbol{\beta}_{[k_0]}\|_2^2}{\sigma^2}\right) $. Note that $({\bf P}_{k_0}-{\bf P}_{k_0-1}){\bf X}_{[k_0]}\boldsymbol{\beta}_{[k_0]}= ({\bf P}_{k_0}-{\bf I}_n+{\bf I}_n-{\bf P}_{k_0-1}){\bf X}_{[k_0]}\boldsymbol{\beta}_{[k_0]}={\bf 0}_n+({\bf I}_n-{\bf P}_{k_0-1}){\bf X}_{[k_0]}\boldsymbol{\beta}_{[k_0]}=({\bf I}_n-{\bf P}_{k_0-1}){\bf x}_{k_0}\boldsymbol{\beta}_{k_0}$.  Hence, $Z_2\sim \chi^2_1\left(\dfrac{\|({\bf I}_n-{\bf P}_{k_0-1}){\bf x}_{k_0}\|_2^2\boldsymbol{\beta}_{k_0}^2}{\sigma^2}\right)$. This proves  b) of Lemma \ref{lemma:basic_distributions}. 
\section*{Appendix B: Proof of Theorem \ref{thm:rrknot}} 
\begin{proof}
Note that $RR(k_0)=\dfrac{Z_1}{Z_1+Z_2}$,  where $Z_1=\|({\bf I}_n-{\bf P}_{k_0}){\bf w}\|_2^2\sim \sigma^2\chi^2_{n-k_0}$ and $Z_2=\|({\bf P}_{k_0}-{\bf P}_{k_0-1}){\bf y}\|_2^2\sim \sigma^2\chi^2_1\left(\dfrac{\|({\bf I}_{n}-{\bf P}_{k_0-1}){\bf x}_{k_0}\|_2^2\boldsymbol{\beta}_{k_0}^2}{\sigma^2}\right)$ as  discussed in Lemma \ref{lemma:basic_distributions}. The proof of Theorem \ref{thm:rrknot} is based on the following lemma. 
\begin{lemma} \label{lemma:noncentral} $\chi^2$ R.V satisfies the following limits\cite{tsp}.\\
a). Let $Z\sim \chi^2_k$ for a fixed $k\in \mathbb{N}$. Then $\sigma^2Z\overset{P}{\rightarrow }0$ as $\sigma^2\rightarrow 0$. \\
b). Let $Z\sim \chi^2_{k}(\lambda/\sigma^2)$ for  fixed $k\in \mathbb{N}$ and fixed $\lambda>0$. Then $\sigma^2 Z\overset{P}{\rightarrow } \lambda$ as $\sigma^2 \rightarrow 0$.
\end{lemma}
It follows directly from Lemma \ref{lemma:noncentral} that $Z_1\overset{P}{\rightarrow} 0 \ \text{as} \ \sigma^2\rightarrow 0$ and $Z_2\overset{P}{\rightarrow} \|({\bf I}_n-{\bf P}_{k_0-1}){\bf x}_{k_0}\|_2^2\boldsymbol{\beta}_{k_0}^2>0 \ \text{as} \ \sigma^2\rightarrow 0$.  Hence, the numerator $Z_1$ in $RR(k_0)=\dfrac{Z_1}{Z_1+Z_2}$ converges in probability to zero, whereas, the denominator $Z_1+Z_2$ converges in probability  to a positive constant $ \|({\bf I}_n-{\bf P}_{k_0-1}){\bf x}_{k_0}\|_2^2\boldsymbol{\beta}_{k_0}^2$. Consequently\footnote{Please note that if $X_n \overset{P}{\rightarrow }c_1$ and $Y_n\overset{P}{\rightarrow } c_2\neq 0$ as $n\rightarrow \infty$, then $X_n/Y_n\overset{P}{\rightarrow } c_1/c_2$. Likewise, $X_n+Y_n\overset{P}{\rightarrow}c_1+c_2$ [Theorem 5.5,\cite{wasserman2013all}].},    $RR(k_0)\overset{P}{\rightarrow } 0$ as $\sigma^2 \rightarrow 0$. 
\end{proof}
\section*{Appendix C: Proof of Theorem \ref{thm:highSNR}}
\begin{proof} The event $\{\hat{k}_{RRT}<k_0\}$ can happen   when either of the events $\mathcal{A}_1=\{\{\exists k<k_0 : RR(k)<\Gamma_{RRT}^{\alpha}(k)\}\cap \{RR(k)>\Gamma_{RRT}^{\alpha}(k),\forall k\geq k_0\}\}$ or $\mathcal{A}_2=\{\{k:RR(k)<\Gamma_{RRT}^{\alpha}(k)\}=\phi\}=\{RR(k)>\Gamma_{RRT}^{\alpha}(k),\ \forall k\}$ is true. $\mathcal{A}_1$ definitely results in $\hat{k}_{RRT}<k_0$, whereas, when $\mathcal{A}_2$ is true,   $\hat{k}_{RRT}=\max\{k:RR(k)<\Gamma_{RRT}^{\alpha_{new}}\}$ in (\ref{alphanew}) can be larger or smaller than $k_0$.   Thus,
\begin{equation}\label{under}
\mathbb{P}_{\mathcal{U}}\leq \mathbb{P}(\mathcal{A}_1)+\mathbb{P}(\mathcal{A}_2).
\end{equation}
Using the bound $\mathbb{P}(\mathcal{B}_1\cap \mathcal{B}_2)\leq \mathbb{P}(\mathcal{B}_1)$ for any two events $\mathcal{B}_1$ and $\mathcal{B}_2$, one can bound 
\begin{equation}
\mathbb{P}(\mathcal{A}_1)\leq \mathbb{P}(\{RR(k_0)>\Gamma_{RRT}^{\alpha}(k_0)\})
\end{equation}
  $RR(k_0)\overset{P}{\rightarrow } 0$ in Theorem \ref{thm:rrknot} implies that 
  \begin{equation}
  \underset{\sigma^2 \rightarrow 0}{\lim}\mathbb{P}(\mathcal{A}_1)\leq \underset{\sigma^2 \rightarrow 0}{\lim}\mathbb{P}(\{RR(k_0)>\Gamma_{RRT}^{\alpha}(k_0)\})=0.
  \end{equation}
   $RR(k_0)\overset{P}{\rightarrow } 0$ also implies that  
 \begin{equation}
 \begin{array}{ll}
 \underset{\sigma^2 \rightarrow 0}{\lim}\mathbb{P}(\mathcal{A}_2)&=\underset{\sigma^2 \rightarrow 0}{\lim}\mathbb{P}(\{RR(k)>\Gamma_{RRT}^{\alpha}(k),\ \forall k\})\\
 &\leq \underset{\sigma^2 \rightarrow 0}{\lim}\mathbb{P}(\{RR(k_0)>\Gamma_{RRT}^{\alpha}({k_0})\})=0.
 \end{array}
 \end{equation}
 Applying these  limits in (\ref{under}) give $\underset{\sigma^2 \rightarrow 0}{\lim }\mathbb{P}_{\mathcal{U}}=0$.
 
 Similarly, the event $\hat{k}_{RRT}>k_0$ can happen either when $\mathcal{A}_3=\{\exists k>k_0:RR(k)<\Gamma_{RRT}^{\alpha}(k)\}$  or when  $\mathcal{A}_2$ is true. When $\mathcal{A}_3$ 
is true, then definitely $\hat{k}_{RRT}>k_0$, whereas, when   $\mathcal{A}_2$ is true, then $\hat{k}_{RRT}$ can be either greater than or smaller than $k_0$. Hence, 
\begin{equation}\label{over}
\mathbb{P}_{\mathcal{O}}\leq \mathbb{P}(\mathcal{A}_3)+\mathbb{P}(\mathcal{A}_2)
\end{equation} 
By Theorem \ref{thm:beta_mos}, $\mathbb{P}(\mathcal{A}_3)=1-\mathbb{P}( \{RR(k)<\Gamma_{RRT}^{\alpha}(k),\forall k>k_0\})\leq \alpha$ for all $\sigma^2>0$.   Applying this along with $\underset{\sigma^2 \rightarrow 0}{\lim}\mathbb{P}(\mathcal{A}_2)=0$   in (\ref{over}) give $\underset{\sigma^2 \rightarrow 0}{\lim }\mathbb{P}_{\mathcal{O}}\leq \alpha$. Note that the events $\{\hat{k}_{RRT}<k_0\}$ and  $\{\hat{k}_{RRT}>k_0\}$ are disjoint and hence 
\begin{equation}
PCS=1-\mathbb{P}(\hat{k}_{RRT}\neq k_0)=1-\mathbb{P}_{\mathcal{U}}-\mathbb{P}_{\mathcal{O}}.
\end{equation}
Thus the limit $\underset{\sigma^2 \rightarrow 0}{\lim }PCS\geq 1-\alpha$ directly follows from the limits on $\mathbb{P}_{\mathcal{U}}$ and $\mathbb{P}_{\mathcal{O}}$.
\end{proof}
\section*{Appendix D: Proof of Theorem \ref{thm:asymptotic_rrt}}
$\Gamma_{RRT}^{\alpha}(k_0)=F^{-1}_{\frac{n-k_0}{2},\frac{1}{2}}(x_n)$, where $x_n=\frac{\alpha}{p}$ is an implicit function of $n$. Depending on the behaviour of $x_n$ as $n \rightarrow \infty$, we consider the following two  cases. 

{\bf Case 1:} $p$ fixed and $\alpha$ fixed. Here $x_n$ is a constant function of $n$ and $k_{lim}=\underset{n \rightarrow \infty}{\lim}k_0/n<1$. Using the limit $\underset{a \rightarrow \infty}{\lim}F^{-1}_{a,b}(x)=1$ for every  fixed $b \in (0,\infty)$ and $x \in(0,1)$ (see proposition 1, \cite{askitis2016asymptotic}), it follows that $\underset{n\rightarrow \infty}{\lim}\Gamma_{RRT}^{\alpha}(k_0)=\underset{n \rightarrow \infty}{\lim}F^{-1}_{\frac{n-k_0}{2},\frac{1}{2}}(x_n)=1$.\\

{\bf Case 2:-} ($p$ fixed, $\alpha \rightarrow 0$), ($p\rightarrow \infty$, $\alpha$ fixed) or ($p\rightarrow \infty$, $\alpha \rightarrow 0$). In all these cases, $x_n\rightarrow 0$ as $n\rightarrow \infty$. Expanding  $F^{-1}_{a,b}(z)$  at $z=0$ using the expansion given in [http://functions.wolfram.com/06.23.06.0001.01] gives
\begin{equation}\label{beta_exp}
\begin{array}{ll}
F^{-1}_{a,b}(z)=\rho(n,1)+\dfrac{b-1}{a+1}\rho(n,2) \\
+\dfrac{(b-1)(a^2+3ab-a+5b-4)}{2(a+1)^2(a+2)}\rho(n,3)
+O(z^{(4/a)})
\end{array}
\end{equation}
for all $a>0$.   We associate $a=\frac{n-k_0}{2}$, $b=1/2$ , $z=x_n$ and $\rho(n,l)=(az{B}(a,b))^{(l/a)}=\left(\frac{\left(\frac{n-k_0}{2}\right)\alpha{B}(\frac{n-k_0}{2},0.5)}{p}\right)^{\frac{2l}{n-k_0}}$ for $l\geq 1$. 
Then $\log(\rho(n,l))$ gives
\begin{equation}\label{log_rho}
\begin{array}{ll}
\log(\rho(n,l))=\frac{2l}{n-k_0}\log\left({\frac{n-k_0}{2p}}\right)
+\frac{2l}{n-k_0}\log({B}(\frac{n-k_0}{2},0.5))\\
+\frac{2l}{n-k_0}\log(\alpha)
\end{array}
\end{equation}
In the limits $n\rightarrow \infty$, $0\leq\underset{n\rightarrow \infty}{\lim} p/n<1$ and $0\leq k_{lim}<1$, $\underset{n \rightarrow \infty}{\lim}\frac{2l}{n-k_0}\log\left({\frac{n-k_0}{2p}}\right)=0$. Using the asymptotic expansion
 ${B}(a,b)=G(b)a^{-b}\left(1-\frac{b(b-1)}{2a}(1+O(\frac{1}{a}))\right)$ as $a \rightarrow \infty$ (given in [http://functions.wolfram.com/06.18.06.0006.01]) in the second term of (\ref{log_rho}) gives
\begin{equation}
\underset{n \rightarrow \infty}{\lim}\frac{2l}{n-k_0}\log\left({B}(\frac{n-k_0}{2},0.5)\right)=0.
\end{equation}
Hence, only the behaviour of $\frac{2l}{n-k_0}\log(\alpha)$ needs to be considered. Now we consider the three cases depending on the behaviour of $\alpha$.

{\bf Case 2.A:-} When $\underset{n \rightarrow \infty}{\lim}\log(\alpha)/n=0$ 
one has $\underset{n \rightarrow \infty}{\lim}\log(\rho(n,l))=0$ which in turn implies that $\underset{n \rightarrow \infty}{\lim}\rho(n,l)=1$ for every $l$. 

{\bf Case 2:-} When $-\infty<\alpha_{lim}=\underset{n \rightarrow \infty}{\lim}\log(\alpha)/n<0$ and $\underset{n \rightarrow \infty}{\lim}\dfrac{k_0}{n}=k_{lim}<1$,   
one has $-\infty<\underset{n \rightarrow \infty}{\lim}\log(\rho(n,l))=(2l\alpha_{lim})/(1-k_{lim})<0$. This in turn implies that $0<\underset{n \rightarrow \infty}{\lim}\rho(n,l)=e^{\dfrac{2l\alpha_{lim}}{1-k_{lim}}}<1$ for every $l$. 

{\bf Case 3:-} When $\underset{n \rightarrow \infty}{\lim}\log(\alpha)/n=-\infty$,  
one has $\underset{n \rightarrow \infty}{\lim}\log(\rho(n,l))=-\infty$ which in turn implies that $\underset{n \rightarrow \infty}{\lim}\rho(n,l)=0$ for every $l$.  

Note that the coefficient of $\rho(n,l)$ in (\ref{beta_exp}) for $l>1$ is asymptotically $1/a\approx 2/(n-k_0)$. Hence, these coefficients decay to zero in the limits $n\rightarrow \infty$ and $0\leq k_{lim}<1$. Consequently, only the $\rho(n,1)$ term in (\ref{beta_exp}) is non zero as $n \rightarrow \infty$. This implies that $\underset{n \rightarrow \infty}{\lim}\Gamma_{RRT}^{\alpha}(k_0)=1$ for Case 2.A,  $0<\underset{n \rightarrow \infty}{\lim}\Gamma_{RRT}^{\alpha}(k_0)=e^{\dfrac{2\alpha_{lim}}{1-k_{lim}}}<1$ for Case 2.B and $\underset{n \rightarrow \infty}{\lim}\Gamma_{RRT}^{\alpha}(k_0)=0$ for Case 2.C. This proves Theorem \ref{thm:asymptotic_rrt}.
\section*{Appendix E. Proof of Theorem \ref{thm:large_sample}}
\begin{proof}
Consider the events $\mathcal{A}_1$, $\mathcal{A}_2$ and $\mathcal{A}_3$ defined in the proof of Theorem \ref{thm:highSNR}. Following the proof of Theorem \ref{thm:highSNR}, one has $\mathbb{P}_{\mathcal{U}}\leq \mathbb{P}(\mathcal{A}_1)+\mathbb{P}(\mathcal{A}_2)$ and $\mathbb{P}_{\mathcal{O}}\leq \mathbb{P}(\mathcal{A}_3)+\mathbb{P}(\mathcal{A}_2)$,  where $\mathbb{P}(\mathcal{A}_1)\leq \mathbb{P}(\{RR(k_0)>\Gamma_{RRT}^{\alpha}(k_0)\})$, $\mathbb{P}(\mathcal{A}_2)\leq \mathbb{P}(\{RR(k_0)>\Gamma_{RRT}^{\alpha}(k_0)\})$ and $\mathbb{P}(\mathcal{A}_3)\leq \alpha,\forall n$. Hence, only the large sample behaviour of $\mathbb{P}(\{RR(k_0)>\Gamma_{RRT}^{\alpha}(k_0)\})$ needs to be analysed. 

Let $Z_1=\|({\bf I}_n-{\bf P}_{k_0}){\bf w}\|_2^2\sim \sigma^2\chi^2_{n-k_0}$ and $Z_2=\|({\bf P}_{k_0}-{\bf P}_{k_0-1}){\bf y}\|_2^2\sim \sigma^2 \chi^2_1\left(\dfrac{\|({\bf I}_n-{\bf P}_{k_0-1}){\bf x}_{k_0}\|_2^2\boldsymbol{\beta}_{k_0}^2}{\sigma^2}\right)$.  Following Lemma \ref{lemma:basic_distributions}, $RR(k_0)=\frac{Z_1}{Z_1+Z_2}$. Hence,
\begin{equation}\label{under_asymptptic}
\begin{array}{ll}
\mathbb{P}(\mathcal{A}_1)=\mathbb{P}\left(\frac{Z_1}{Z_1+Z_2}>\Gamma_{RRT}^{\alpha}(k_0)\right)
=\mathbb{P}\left(\frac{1-\Gamma_{RRT}^{\alpha}(k_0)}{\Gamma_{RRT}^{\alpha}(k_0)}\frac{Z_1}{n\sigma^2}>\frac{Z_2}{n\sigma^2}\right)
\end{array}
\end{equation}
The large sample behaviour of  chi squared R.Vs are characterized  in the following lemma.
\begin{lemma}\label{lemma:noncentral2}
Chi squared R.Vs satisfy the following limits.\\
A1). Let $Z\sim {\chi^2_l}$, then $Z/l \overset{P}{\rightarrow} 1$ as $l \rightarrow \infty$. \\
A2). Let $Z\sim \chi^2_k(M l)$ for a fixed $k$ and $M>0$, then $Z/l\overset{P}{\rightarrow} M$ as $l \rightarrow \infty$. \cite{eefenumeration}
\end{lemma} 
By Theorem 4, $\alpha_{\lim}=0$ implies that $\Gamma_{RRT}^{\alpha}(k_0)\rightarrow 1$ and $\dfrac{1-\Gamma_{RRT}^{\alpha}(k_0)}{\Gamma_{RRT}^{\alpha}(k_0)}\rightarrow 0$ as $n \rightarrow \infty$.  A1) of Lemma \ref{lemma:noncentral2} and $0\leq k_{lim}<1$ imply $\dfrac{Z_1}{n\sigma^2}=\dfrac{Z_1}{(n-k_0)\sigma^2}\dfrac{n-k_0}{n}\overset{P}{\rightarrow }1-k_{lim}$. Combining these limits,   L.H.S in (\ref{under_asymptptic}), i.e., $\dfrac{1-\Gamma_{RRT}^{\alpha}(k_0)}{\Gamma_{RRT}^{\alpha}(k_0)}\dfrac{Z_1}{n\sigma^2}\overset{P}{\rightarrow} 0$ as $n \rightarrow \infty$. Next we consider the behaviour of $\dfrac{Z_2}{n\sigma^2}$. Let $\tilde{Z}_2\sim \chi^2_1(M_1n)$ be some R.V. Then 
\begin{equation}\label{gargon}
\mathbb{P}\left(\frac{Z_2}{n\sigma^2}>\frac{M_1}{2}\right)\geq \mathbb{P}\left(\frac{\tilde{Z}_2}{n}>\frac{M_1}{2}\right), \ \forall n>n_0.
\end{equation}
Eq.\ref{gargon} follows from the monotonicity of $\chi^2_k(\lambda)$ w.r.t $\lambda$ and the fact that the noncentrality parameter in $Z_2$ satisfies $\dfrac{\|({\bf I}_n-{\bf P}_{k_0-1}){\bf x}_{k_0}\|_2^2\boldsymbol{\beta}_{k_0}^2}{\sigma^2}\geq M_1n$ for all $n>n_0$. A2) of Lemma \ref{lemma:noncentral2} implies that $\tilde{Z}_2/n\overset{P}{\rightarrow }M_1$ as $n \rightarrow \infty$. This implies that
\begin{equation}
\underset{n \rightarrow \infty}{\lim}\mathbb{P}\left(\frac{\tilde{Z}_2}{n}>\frac{M_1}{2}\right)=1\  \text{and} \ \ \underset{n \rightarrow \infty}{\lim} \mathbb{P}\left(\frac{Z_2}{n\sigma^2}>\frac{M_1}{2}\right)=1.
\end{equation}
Since the L.H.S of (\ref{under_asymptptic}) converges to zero and R.H.S is bounded away from zero with probability one, it is true that   $\underset{n \rightarrow \infty}{\lim}\mathbb{P}(\mathcal{A}_1)=0$. Similarly, $\underset{n \rightarrow \infty}{\lim}\mathbb{P}(\mathcal{A}_2)=0$.

Note that the limits derived so far assumed only $\alpha_{lim}=0$. Hence, as long as $\alpha_{lim}=0$, it is true that $\underset{n \rightarrow \infty}{\lim}\mathbb{P}_{\mathcal{U}}= 0$ and $\underset{n \rightarrow \infty}{\lim}\mathbb{P}_{\mathcal{O}}\leq \alpha$. Since,  $\alpha_{lim}=0$ for fixed $0\leq \alpha\leq 1$, this proves R2) of Theorem \ref{thm:large_sample}. Once $\underset{n \rightarrow \infty}{\lim}\alpha=0$ is also true, then $\underset{n \rightarrow \infty}{\lim}\mathbb{P}_{\mathcal{O}}= 0$ and $\underset{n \rightarrow \infty}{\lim}PCS=1-\underset{n \rightarrow \infty}{\lim}\mathbb{P}_{\mathcal{O}}-\underset{n \rightarrow \infty}{\lim}\mathbb{P}_{\mathcal{U}}=1$. This proves R1) of Theorem \ref{thm:large_sample}. 
 \end{proof}

\bibliographystyle{IEEEtran}
\bibliography{compressive}

\end{document}

%% file: notation.tex
\newcommand{\upperRomannumeral}[1]{\uppercase\expandafter{\romannumeral#1}}

\newtheorem{lemma}{Lemma}{}
  \newtheorem{thm}{Theorem}

%% file: residual_ratio_thresholding.bbl
\begin{thebibliography}{10}
\providecommand{\url}[1]{#1}
\csname url@samestyle\endcsname
\providecommand{\newblock}{\relax}
\providecommand{\bibinfo}[2]{#2}
\providecommand{\BIBentrySTDinterwordspacing}{\spaceskip=0pt\relax}
\providecommand{\BIBentryALTinterwordstretchfactor}{4}
\providecommand{\BIBentryALTinterwordspacing}{\spaceskip=\fontdimen2\font plus
\BIBentryALTinterwordstretchfactor\fontdimen3\font minus
  \fontdimen4\font\relax}
\providecommand{\BIBforeignlanguage}[2]{{%
\expandafter\ifx\csname l@#1\endcsname\relax
\typeout{** WARNING: IEEEtran.bst: No hyphenation pattern has been}%
\typeout{** loaded for the language `#1'. Using the pattern for}%
\typeout{** the default language instead.}%
\else
\language=\csname l@#1\endcsname
\fi
#2}}
\providecommand{\BIBdecl}{\relax}
\BIBdecl

\bibitem{stoica2004model}
P.~Stoica and Y.~Selen, ``Model-order selection: A review of information
  criterion rules,'' \emph{IEEE Signal Process. Mag.}, vol.~21, no.~4, pp.
  36--47, July 2004.

\bibitem{raghavendra2005improving}
M.~Raghavendra and K.~Giridhar, ``Improving channel estimation in {OFDM}
  systems for sparse multipath channels,'' \emph{IEEE Signal Process. Lett.},
  vol.~12, no.~1, pp. 52--55, 2005.

\bibitem{tomasoni2013efficient}
A.~Tomasoni, D.~Gatti, S.~Bellini, M.~Ferrari, and M.~Siti, ``Efficient {OFDM}
  channel estimation via an information criterion,'' \emph{IEEE trans. wireless
  commun.}, vol.~12, no.~3, pp. 1352--1362, 2013.

\bibitem{filter_design}
J.~Botts, J.~Escolano, and N.~Xiang, ``Design of {IIR} filters with {B}ayesian
  model selection and parameter estimation,'' \emph{IEEE Trans. Audio, Speech,
  Language Process.}, vol.~21, no.~3, pp. 669--674, March 2013.

\bibitem{schmidt2011estimating}
D.~F. Schmidt and E.~Makalic, ``Estimating the order of an autoregressive model
  using normalized maximum likelihood,'' \emph{IEEE Trans. Signal Process.},
  vol.~59, no.~2, pp. 479--487, 2011.

\bibitem{stoica2004information}
P.~Stoica, Y.~Selen, and J.~Li, ``On information criteria and the generalized
  likelihood ratio test of model order selection,'' \emph{IEEE Signal Process.
  Lett.}, vol.~11, no.~10, pp. 794--797, 2004.

\bibitem{stoica2013model}
P.~Stoica and P.~Babu, ``Model order estimation via penalizing adaptively the
  likelihood ({PAL}),'' \emph{Signal Processing}, vol.~93, no.~11, pp. 2865 --
  2871, 2013.

\bibitem{rissanen2000mdl}
J.~Rissanen, ``{MDL} denoising,'' \emph{IEEE Trans. Inf. Theory}, vol.~46,
  no.~7, pp. 2537--2543, Nov 2000.

\bibitem{stoica2012proper}
P.~Stoica and P.~Babu, ``On the proper forms of {BIC} for model order
  selection,'' \emph{IEEE Trans. Signal Process.}, vol.~60, no.~9, pp.
  4956--4961, Sept 2012.

\bibitem{nielsen2013bayesian}
J.~K. Nielsen, M.~G. Christensen, and S.~H. Jensen, ``Bayesian model comparison
  and the {BIC} for regression models,'' in \emph{Proc. IEEE ICAASP}.\hskip 1em
  plus 0.5em minus 0.4em\relax IEEE, 2013, pp. 6362--6366.

\bibitem{EEFPDF}
S.~Kay, Q.~Ding, B.~Tang, and H.~He, ``Probability density function estimation
  using the {EEF} with application to subset/feature selection,'' \emph{IEEE
  Trans. Signal Process.}, vol.~64, no.~3, pp. 641--651, Feb 2016.

\bibitem{ding2011inconsistency}
Q.~Ding and S.~Kay, ``Inconsistency of the {MDL}: {O}n the performance of model
  order selection criteria with increasing signal-to-noise ratio,'' \emph{IEEE
  Trans. Signal Process.}, vol.~59, no.~5, pp. 1959--1969, May 2011.

\bibitem{tsp}
S.~Kallummil and S.~Kalyani, ``High {SNR} consistent linear model order
  selection and subset selection,'' \emph{IEEE Trans. Signal Process.},
  vol.~64, no.~16, pp. 4307--4322, Aug 2016.

\bibitem{nishii1988maximum}
R.~Nishii, ``Maximum likelihood principle and model selection when the true
  model is unspecified,'' \emph{Journal of Multivariate Analysis}, vol.~27,
  no.~2, pp. 392--403, 1988.

\bibitem{asymptotic_map}
P.~M. Djuric, ``Asymptotic {MAP} criteria for model selection,'' \emph{IEEE
  Trans. on Signal Process.}, vol.~46, no.~10, pp. 2726--2735, Oct 1998.

\bibitem{rao1989strongly}
R.~RAO and Y.~WU, ``A strongly consistent procedure for model selection in a
  regression problem,'' \emph{Biometrika}, vol.~76, no.~2, pp. 369--374, 1989.

\bibitem{shao1997asymptotic}
J.~Shao, ``An asymptotic theory for linear model selection,'' \emph{Stat.
  Sin.}, vol.~7, no.~2, pp. 221--242, 1997.

\bibitem{zheng1995consistent}
X.~Zheng and W.-Y. Loh, ``Consistent variable selection in linear models,''
  \emph{J. Amer. Statist. Assoc.}, vol.~90, no. 429, pp. 151--156, 1995.

\bibitem{minimal}
E.~Gassiat and R.~van Handel, ``Consistent order estimation and minimal
  penalties,'' \emph{IEEE Trans. Inf. Theory}, vol.~59, no.~2, pp. 1115--1128,
  Feb 2013.

\bibitem{schmidt2012consistency}
D.~Schmidt and E.~Makalic, ``The consistency of {MDL} for linear regression
  models with increasing signal-to-noise ratio,'' \emph{IEEE Trans. Signal
  Process.}, vol.~60, no.~3, pp. 1508--1510, March 2012.

\bibitem{designITC}
A.~Mariani, A.~Giorgetti, and M.~Chiani, ``Model order selection based on
  information theoretic criteria: {D}esign of the penalty,'' \emph{IEEE Trans.
  Signal Process.}, vol.~63, no.~11, pp. 2779--2789, June 2015.

\bibitem{lu2013generalized}
Z.~Lu and A.~Zoubir, ``Generalized {B}ayesian information criterion for source
  enumeration in array processing,'' \emph{IEEE Trans. Signal Process.},
  vol.~61, no.~6, pp. 1470--1480, March 2013.

\bibitem{haddadi2010statistical}
F.~Haddadi, M.~Malek-Mohammadi, M.~Nayebi, and M.~Aref, ``Statistical
  performance analysis of {MDL} source enumeration in array processing,''
  \emph{IEEE Trans. Signal Process.}, vol.~58, no.~1, pp. 452--457, Jan 2010.

\bibitem{nielsen2014model}
J.~K. Nielsen, M.~G. Christensen, and S.~H. Jensen, ``Model selection and
  comparison for independents sinusoids,'' in \emph{Proc. ICAASP}.\hskip 1em
  plus 0.5em minus 0.4em\relax IEEE, 2014, pp. 1891--1895.

\bibitem{cai2011orthogonal}
T.~Cai and L.~Wang, ``Orthogonal matching pursuit for sparse signal recovery
  with noise,'' \emph{IEEE Trans. Inf. Theory}, vol.~57, no.~7, pp. 4680--4688,
  July 2011.

\bibitem{decoding_candes}
E.~J. Candes and T.~Tao, ``Decoding by linear programming,'' \emph{IEEE Trans.
  Inf. Theory}, vol.~51, no.~12, pp. 4203--4215, Dec 2005.

\bibitem{yanai2011projection}
H.~Yanai, K.~Takeuchi, and Y.~Takane, \emph{Projection Matrices}.\hskip 1em
  plus 0.5em minus 0.4em\relax Springer, 2011.

\bibitem{ravishanker2001first}
N.~Ravishanker and D.~K. Dey, \emph{A first course in linear model
  theory}.\hskip 1em plus 0.5em minus 0.4em\relax CRC Press, 2001.

\bibitem{wasserman2013all}
L.~Wasserman, \emph{All of statistics: a concise course in statistical
  inference}.\hskip 1em plus 0.5em minus 0.4em\relax Springer Science \&
  Business Media, 2013.

\bibitem{askitis2016asymptotic}
D.~Askitis, ``Asymptotic expansions of the inverse of the beta distribution,''
  \emph{arXiv preprint arXiv:1611.03573}, 2016.

\bibitem{eefenumeration}
C.~Xu and S.~Kay, ``Source enumeration via the {EEF} criterion,'' \emph{IEEE
  Signal Process. Lett.}, vol.~15, pp. 569--572, 2008.

\end{thebibliography}
